\documentclass[11pt]{article}   \usepackage{ASA_preamble}
\DeclareMathSizes{11}{10}{8}{6} 

\title{Localisation without supersymmetry: towards exact results from Dirac structures in 3D $N=0$ gauge theory}

\date{\today}

\author{Alex S.~Arvanitakis\thanks{Email: \texttt{alex.s.arvanitakis@vub.be}},\,  Dimitri Kanakaris\thanks{Email: \texttt{dimitri.kanakaris.decavel@outlook.com}}}
\affil{\small
\textit{Theoretische Natuurkunde, Vrije Universiteit Brussel,
and the International Solvay Institutes,
Pleinlaan 2, B-1050 Brussels, Belgium}
}

\author{DKD }

\newcommand{\trd}{Third Way}
\newcommand{\bbA}{\mathbb A}
\newcommand{\tfg}{\tilde{\mathfrak g}}
\newcommand{\fd}{\mathfrak{d}}
\newcommand{\tA}{\tilde A}
\newcommand{\tT}{\tilde T}
\newcommand{\tf}{\tilde f}
\newcommand{\bbg}{\mathbbm g}

\newcommand{\extd}{\mathrm{d}}
\newcommand{\del}{\partial}
\newcommand{\bbD}{\mathbb D}
\newcommand{\bbF}{\mathbb F}

\newcommand{\diag}{\operatorname{diag}}

\newcommand{\CS}{{\text{CS}}}
\newcommand{\BRST}{{\text{BRST}}}
\newcommand{\loc}{{\text{loc}}}
\newcommand{\Or}{\mathcal{O}}

\newcommand{\gD}{\mathcal{D}}
\newcommand{\gsD}{\slashed{\gD}}
\newcommand{\gE}{\mathcal{E}}

\newcommand{\gM}{\mathcal{M}}
\newcommand{\gR}{\mathcal{R}}
\newcommand{\bc}{\bar{c}}

\newcommand{\hlambda}{\hat\lambda}
\newcommand{\hzeta}{{\hat\zeta}}

\newcommand{\hdelta}{\hat\delta}
\newcommand{\sF}{\slashed F}
\newcommand{\bbsF}{\slashed\bbF}
\newcommand{\bbhlambda}{\hat\bblambda}
\newcommand{\bbsA}{\slashed\bbA}

\newcommand{\Cj}{\mathcal{C}}

\newcommand{\zetah}{\hat\zeta}

\newcommand{\nablas}{\slashed{\nabla}}

\newcommand{\nablasl}{\overset{\leftarrow}{\nablas}}

\newcommand{\Dg}{\mathcal{D}}
\newcommand{\Dgs}{\slashed{\Dg}}

\newcommand{\pder}[2]{{\frac{\del #1}{\del #2}}}

\newcommand{\Ber}{\operatorname{Ber}}
\newcommand{\Hess}{\operatorname{Hess}}
\newcommand{\Pf}{\operatorname{Pf}}
\newcommand{\bos}{{\text{bos}}}

\newcommand{\idop}{{\mathbbm{1}}}

\definecolor{britishracinggreen}{rgb}{0.0, 0.26, 0.15}
\definecolor{lemonchiffon}{rgb}{1.0, 0.98, 0.8}
\definecolor{mediumseagreen}{rgb}{0.24, 0.7, 0.44}
\definecolor{parisgreen}{rgb}{0.31, 0.78, 0.47}

\begin{document}
\titlepage
\maketitle
\thispagestyle{empty}
\begin{abstract}
We show, by introducing purely auxiliary gluinos and scalars, that the quantum path integral for a class of 3D interacting non-supersymmetric gauge theories localises. The theories in this class all admit a `Manin gauge theory' formulation, that we introduce; it is obtained by enhancing the gauge algebra of the theory to a Dirac structure inside a Manin pair. This machinery allows us to do  localisation computations for every theory in this class at once, including for 3D Yang-Mills theory, and for its Third Way deformation; the latter calculation casts the Third Way path integral into an almost 1-loop exact form.

\end{abstract}
\pagebreak 
\thispagestyle{empty}
\tableofcontents

\clearpage
\setcounter{page}{1}

\section{Introduction}
Our elders sometimes advise, ``one idea per paper''.  We will disappoint our elders, for this paper contains two ideas:
\begin{enumerate}
    \item {\bf Manin theory:} a new formulation of three-dimensional Yang-Mills-esque gauge theories, which encodes distinct models with gauge algebra $\fg$ depending on how $\fg$ is chosen as a \emph{Dirac structure} inside a bigger Lie algebra $\fd$, called its \emph{double}; and
    \item {\bf Evanescent localisation:} the addition of purely auxiliary degrees of freedom --- auxiliary gluinos and scalar fields for this paper --- such that the path integral that calculates certain expectation values (e.g.~the partition function)  reduces to an equivalent integral over a smaller space, such that the theory with the auxiliary degrees of freedom is completely equivalent to the original theory.
\end{enumerate}
What we do in this paper is combine ideas 1.~and 2.~to produce a localisation calculation for large classes of Manin gauge theories at once. The motivation for this is that these theories are \textbf{interacting, non-supersymmetric,  non-topological gauge theories} (i.e.~they possess local degrees of freedom) which are not known to be integrable or otherwise amenable to exact quantum calculations: for instance, this class of theories contains (a theory equivalent to) ordinary 3D $N=0$ Yang-Mills theory.

Although the machinery of Dirac structures, Manin pairs, and Lie quasi-bialgebras, that we will employ, might be obscure to some physicists, it has previously made numerous appearances in string theory, including: in Poisson-Lie T-duality \cite{Klimcik:1995ux,Klimcik:1995jn}, in 2D integrable \cite{Sfetsos:2015nya,Vicedo:2015pna} and topological \cite{Getzler:1993fs,Kotov:2004wz,Chatzistavrakidis:2022wdd} sigma models, in the realisation of symmetries on string and brane worldvolumes \cite{Alekseev:2004np,Arvanitakis:2021wkt}, and in the geometry of flux compactifications \cite{gualtieri2011generalized,Tennyson:2021qwl}, to give an incomplete list. It is therefore somewhat surprising that the same ideas are fruitful in the non-stringy context of the current paper.

The localisation technique we employ is the usual supersymmetric localisation for gauge theories, which was introduced by Pestun \cite{Pestun:2007rz} for 4D super Yang-Mills theory and was subsequently employed by Kapustin, Willett, and Yaakov \cite{kapustin2010exact} to calculate expectation values of supersymmetric Wilson loops in various 3D $N=2$ gauge theories; we largely follow the latter. (We also benefitted from the treatment in references \cite{Kallen:2011ny,Fan_2019}.) The $N=2$ `supersymmetry' we employ is realised via purely auxiliary gluinos and scalar fields, and leads to evanescent localisation in the above sense. To highlight this fact, and to distinguish between our `supersymmetry' and conventional supersymmetry, we will call the transformations we employ \textbf{evanescent supersymmetries}. The key point in our argument is the construction of $N=2$ evanescent-supersymmetric Manin gauge theory which is amenable to supersymmetric localisation techniques.

It is worth pointing out early that while the evanescent supersymmetries themselves are just deformed versions of garden-variety supersymmetries, their algebra is typically completely different, for the following reason: evanescent supersymmetries are necessarily vanishing on-shell (``trivial'' in the terminology of Henneaux and Teitelboim \cite{henneaux1992quantization}), and their anticommutators are necessarily trivial  (as we will see); however, spacetime translations are usually nonzero on-shell, hence evanescent SUSYs are not expected to square to translations, and indeed they do not for $N=2$ evanescent-SUSY Manin theory. An exception to this trend is pure $N=2$ Chern-Simons theory which is in fact evanescent-supersymmetric in our sense, as we review. 

Moreover, we establish various basic facts about Manin theory in order to demonstrate that the theory both passes basic consistency checks and can be localised. Among other results, we derive its hamiltonian formulation and demonstrate that energy is bounded below (section \ref{sec:hamiltonian}), we show that the theory enjoys parity invariance and reflection positivity (in Euclidean signature, section \ref{secParity}), and we  demonstrate that the theory can be supersymmetrised on curved compact backgrounds (section \ref{sec:SUSYcurved}).

Finally, just to assure the reader that we are not talking about the empty set, we take care to display examples of theories which admit a Manin theory description (section \ref{sec:examples}): these include Yang-Mills theory, a deformation thereof called the ``Third Way'' theory introduced a while ago by Sevrin, Townsend, and the first author \cite{Arvanitakis:2015oga}, and a few others. (In fact the localisation results in this paper generalise and expand the results derived for the Third Way theory in the second author's Master's thesis \cite{Decavel:2023rmi}.) We also point out that solutions to the modified classical Yang-Baxter equation, or to a generalisation thereof we introduce in this paper \eqref{eq:modifiedmodifiedCYBE}, \emph{always} give rise to Manin theories.

\section{Manin gauge theory}
Our theory --- which we will call ``Manin theory'' for brevity and to avoid self-aggrandisement\footnote{And also because ``Dirac theory'' and ``Dirac gauge theory'' are very much taken!} --- is defined by the following mass deformation of a Chern-Simons action:
\be
\label{eq:ManinLag}
S[\bbA]=\int_{M_3}\eta\bigg[k\Big(\frac{1}{2}\bbA\dr\bbA +\frac{1}{3}\bbA^3\Big) + \frac{1}{2}  g^2 \bbA\star M\bbA\bigg]\,,
\ee
The lagrangian is defined on any three-dimensional manifold $M_3$ with Hodge star $\star$. It is completely specified by the following data:
\begin{itemize}
    \item a dimensionless real constant $k$;
    \item  a constant $g^2$ with units of mass;
    \item a \emph{Manin pair} $(\frak d,\frak g,\eta)$, with the 1-form field $\bbA$ taking values in the Lie algebra $\fd$, and where $\eta$ is the invariant inner product for $\fd$, with $\fg\into \fd$ being the gauge algebra;
    \item and an operator $M:\fd\to \fd$  obeying identities \eqref{eq:Midentities} below.
\end{itemize} As we will see later, whenever the theory is equivalent to Yang-Mills, $g^2$ is proportional to the Yang-Mills coupling. We will also see that the operator $M$ is often completely specified by the Manin pair data. We now summarise what a Manin pair is and what the identities to be satisfied by $M$ are.

\paragraph{Lightning definition of \emph{Manin pairs}, \emph{Lie quasibialgebras,} and \emph{Dirac structures}.} These notions essentially contain the same information but from different perspectives. 
\begin{itemize}
    \item A Lie quasibialgebra is a real Lie algebra $\fg$ with commutation relations $[T_a,T_b]=f_{ab}{}^c T_c$  which is additionally endowed with  objects $\tf^{ab}{}_c$ and $\tilde h^{abc}$, such that $(f_{ab}{}^c,\tf^{ab}{}_c, \tilde h^{abc})$ define the structure constants of a $(2\dim\fg)$-dimensional Lie algebra $\frak d$ as follows:
\be
\label{eq:manindoublecommutationrelations}
[T_a,T_b]=f_{ab}{}^c T_c\,,\quad [\tT^a,\tT^b]= \tf^{ab}{}_c \tT^c+ \tilde h^{abc} T_c\,,\quad [T_a,\tT^b]=\tf^{bc}{}_a T_c -f_{ac}{}^b \tT^c\,,
\ee
where $T_a$ and $\tilde T^a$ collectively form a basis of $\fd$. The $\tf^{ab}{}_c$ and $\tilde h^{abc}$ are totally antisymmetric in their upper indices and must obey identities amongst themselves implied by the Jacobi identity for $\fd$. With these structure constants for $\fd$ we find that the split-signature inner product with nonvanishing matrix entries $\eta(T_a,\tT^b)=\delta^a_b$ is $\ad \fd$-invariant ($\eta([\mathbbm x,\mathbbm y],\mathbbm z)=\eta(\mathbbm x,[\mathbbm y,\mathbbm z])$ for $\mathbbm{x,y,z}\in\fd$).

$\fd$ is called the \emph{double} of the Lie quasibialgebra $\fg$ equipped with $\tilde f,\tilde h$ as above, following Bangoura and Kosmann-Schwarzbach \cite{bangoura1993double}.
\item A \emph{Manin pair} is defined to be the triple $(\fd,\fg,\eta)$ where $(\fd,\eta)$ is a Lie algebra with a split-signature invariant inner product $\eta$ and a  subalgebra $\fg\into\fd$ where $\eta$ vanishes: $\eta|_{\fg}=0$.

Upon choosing any complementary isotropic vector space $\tfg$ to $\fg$ in $\fd$, $\fg$ acquires the structure of a Lie quasibialgebra, with commutation relations as above. Therefore there are multiple Lie quasibialgebras corresponding to each Manin pair. We discuss the \emph{twist} equivalences between such Lie quasibialgebras later.

Note that $\tfg$ need not be a Lie algebra itself, i.e.~we need not have $\tilde h=0$. (In this last case, $\fg$ is called a \emph{Lie bialgebra}, and $(\fd,\fg,\tfg)$ form a \emph{Manin triple}.)

\item In the above scenarios, $\fg$ is a \emph{Dirac structure} for the specific double $\fd$; this is a maximal isotropic --- i.e.~lagrangian --- subalgebra of $\fd$, and we will also refer to $\fg$ as `the lagrangian subalgebra' of $\fd$ in this context. (The notion of Dirac structure was introduced in a more general context in \cite{courant1990dirac}.)
\end{itemize}

\paragraph{The $M$ operator.} This is a linear map $M:\frak d\to \frak d$ with
\begin{subequations}
\label{eq:Midentities}
\begin{alignat}{2}
M\frak{g}&=0\,;\\
\eta(M\mathbbm x,\mathbbm y)&=\eta(\mathbbm x,M\mathbbm y)\,;\\
\eta(M[x,\mathbbm y],\mathbbm z)+\eta(M\mathbbm y, [x,\mathbbm z])&=0\,,
\end{alignat}
\end{subequations}
where $\mathbbm{x,y,z}\in \fd$ and $x\in \fg$. In the basis $(T_a,\tT^a)$ with structure constants as in \eqref{eq:manindoublecommutationrelations} one can prove easily that
\be
MT_a=0\,,\qquad  M\tilde T^a=M^{ab}T_b\,,
\ee for a symmetric matrix $M^{ab}$. Moreover we will assume $M$ is \textbf{nondegenerate} in the sense that the form $\eta(M\bullet,\bullet)$ restricted to an isotropic complement, $\tfg$, of $\fg$ in $\fd=\fg+\tfg$ is nondegenerate; in terms of $M^{ab}$ this is the condition that $M^{ab}$ admits an inverse. We will denote that inverse by $M_{ab}$.

In the basis \eqref{eq:manindoublecommutationrelations} for $\fd$ we may split $\bbA\equiv A^a T_a+\tA_a \tT^a$ so the action reads
\be
\label{eq:ManinLagExplicit}
\begin{split}
S=\frac{1}{2}\int\bigg[ k\Big(\tilde A_a \dr A^a + A^a\dr \tilde A_a+ f_{bc}{}^a \tilde A_a A^b A^c +  \tilde f^{bc}{}_aA^a \tilde A_b\tilde A_c+\frac{1}{3} \tilde h^{abc}\tilde A_a\tilde A_b\tilde A_c \Big)\\
+g^2\tilde A_a \star M^{ab}\tilde A_b\bigg]\,.
\end{split}
\ee
Although this \emph{form} of the action depends on the choice of complement $\tfg$ of $\fg$ in $\fd=\fg+\tfg$, we emphasise that the action does not depend on $\tfg$, as is manifest from \eqref{eq:ManinLag} and \eqref{eq:Midentities}. In other words, the action only depends on the Manin pair, and not on the specific Lie quasibialgebra whose data appears in \eqref{eq:ManinLagExplicit}.

\paragraph{Gauge invariance.} The action \eqref{eq:ManinLag} (or equivalently \eqref{eq:ManinLagExplicit}) is invariant (up to boundary terms) under a subset of the gauge transformations of Chern-Simons theory with gauge algebra $\fd$. In fact it is gauge invariant under gauge transformations valued in the Lagrangian subalgebra $\fg$. Explicitly, those are
\be
\label{eq:infinitesimalgaugetransfs}
\delta \bbA= \dr \Lambda+[\bbA,\Lambda]\,,\qquad \Lambda\in \cin(M_3)\times \fg\,,
\ee
where the bracket is given in \eqref{eq:manindoublecommutationrelations}. The proof of gauge invariance is trivial and involves checking the invariance of the mass term  $\eta(\bbA \star M\bbA)$ using the identities \eqref{eq:Midentities}. One may similarly show gauge invariance under finite gauge transformations with gauge parameter $g:M_3\to G$, $G$ being the Lie group with Lie algebra $\fg$.

In terms of the split $\bbA=A+\tA$ into $\fg$- and $\tfg$-valued 1-forms as used to arrive at \eqref{eq:ManinLagExplicit}, the commutation relations \eqref{eq:manindoublecommutationrelations} show that $A$ transforms as a \textbf{gauge field} for $\fg$ while $\tA$ transforms as a \textbf{matter field}.

\paragraph{Remarks}\begin{enumerate}
    \item Any admissible $M$ (that obeys \eqref{eq:Midentities}) may be rescaled arbitrarily, which corresponds to adjusting the coupling $g^2$. Moreover, if $\fg$ is a simple Lie algebra then the form $\eta(M\bullet,\bullet)$ and $M$ are both  determined (via Schur's lemma) up to scale in terms of the inverse of the Killing form for $\fg$. Explicitly, we have 
    \be
    M^{ab}\propto \kappa^{ab}
    \ee
    if $\kappa$ denotes the Killing form on $\fg$.
    \item The algebra $\fg$ can be compact, however its double $\fd$ is typically not compact because $\eta$ has split signature. In particular $\fd$ is never compact when $\fd$ is simple. We will see examples of both.
    \item One might expect a quantisation condition for the dimensionless constant $k$ since it appears as a Chern-Simons level. Ignoring, even, the mass term that breaks some of the gauge symmetry, there is usually no quantisation condition: $\eta((\bbg\inv\dr\bbg)^3)$  can be globally exact, where $\bbg:M_3\to \mathbb D$ is a finite gauge parameter, and $\mathbb D$ is a Lie group with Lie algebra $\fd$. (It is indeed exact e.g.~if the brackets \eqref{eq:manindoublecommutationrelations} close on $\tfg$ so it is a Lie algebra while at the same time $\mathbb D$ is diffeomorphic to $G\times \tilde G$, $\tilde G$ being a Lie group whose Lie algebra is $\tfg$, see e.g.~\cite[section 3.2]{Arvanitakis:2021lwo}.)
\end{enumerate}

\subsection{Examples}
\label{sec:examples}
Before studying the class of Manin theories as a whole, we demonstrate that most known 3D $N=0$ gauge theories belong in this class, and point out examples of Manin pairs that give new 3D gauge theories.

Since one needs to specify not just the gauge algebra $\fg$ but the Lie group $G$ in order to fully determine a gauge theory, we will be denoting the Manin pairs $(\fd,\fg)$ via their corresponding Lie groups $(\mathbb D,G)$, where the gauge group $G$ is a maximally isotropic (with respect to $\eta$) subgroup of $\mathbb D$. 

\subsubsection{Yang-Mills as a Manin theory}\label{sec:YM}
\label{Yang-Mills as a Manin theory}
Here we take $\mathbb D=T^\star G$, the cotangent bundle of the gauge group. This is in fact a Lie group; it is the semidirect product $G\ltimes \mathfrak g^\star$ of $G$ with its coadjoint representation. If $T_a$ denote generators of $\mathfrak{g}$ and $\tilde T^a$ those of $\mathfrak{g}^\star\equiv \tfg$, the Lie algebra $\fd$ of $\mathbb D=T^\star G$ takes the form \eqref{eq:manindoublecommutationrelations} with $\tilde f=\tilde h=0$. Explicitly,
    \be
    \label{eq:YMdoublealgebra}
    [T_a,T_b]=f_{ab}{}^c T_c\,,\qquad [\tilde T^a,T_b]=f_{bc}{}^a \tilde T^c\,,\qquad [\tilde T^a,\tilde T^b]=0\,,
    \ee
so that the action (after an integration by parts) reads
\be
\label{first order YM}
\begin{split}
S=\int\bigg[ k\Big(\tilde A_a (\dr A^a + \tfrac{1}{2}f_{bc}{}^a A^b A^c) \Big)+\frac{1}{2}g^2\tilde A_a \star M^{ab}\tilde A_b\bigg]\,.
\end{split}
\ee
For $G$ a compact simple group, as we remarked above, $M^{ab}$ must be proportional to the inverse Killing form $\kappa^{ab}$.  We then recognise the above as the first-order formulation of Yang-Mills theory. Indeed upon elimination of $\tilde A$ we obtain
\be
S=-\frac{k^2}{2g^2}\int M_{ab} F^a \star F^b\qquad (F\equiv \dr A+A^2)\,,
\ee
which is the Yang-Mills action with coupling proportional to $g^2/k^2$.

We observe here that $k$ clearly need not obey a quantisation condition in this case.

\subsubsection{Freedman-Townsend theory in 3D}
\label{sec:freedmantownsend}
For the above example, since $\tilde h=0$ we notice that $\fg$ is in fact a Lie bialgebra. Therefore we can swap the roles of $\fg$ and its complement $\fg^\star\equiv \tfg$ and consider the Manin pair of $\fd=\fg\ltimes \fg^\star$ and $\fg^\star$ (which we interpret as the abelian Lie algebra of dimension $\dim \fg$). To do this we need to replace $M$ by a new operator $\tilde M$ appropriate to this Manin pair; in particular, due to \eqref{eq:Midentities}, it must annihilate $\fg^\star$. Its nonvanishing matrix elements will therefore be $\tilde M T_a=\tilde M_{ab}\tilde T^b$.

Therefore the action reads (after integration by parts)
\be
\begin{split}
S=\int\bigg[ k\Big(\tilde A_a (\dr A^a + \tfrac{1}{2}f_{bc}{}^a A^b A^c) \Big)+\frac{1}{2}g^2 A^a \star \tilde M_{ab} A^b\bigg]\,.
\end{split}
\ee
We recognise this as the Freedman-Townsend gauge theory action, or rather a three-dimensional version thereof \cite[formula (2.9)]{Freedman:1980us} (at least when $\tilde M$ is proportional to the Killing form).

The gauge transformations \eqref{eq:infinitesimalgaugetransfs} indeed reproduce the ones of Freedman-Townsend theory: we need to replace the gauge parameter $\Lambda=\Lambda^a T_a$ that takes values in $\fg$ with a gauge parameter $\tilde \Lambda=\tilde \Lambda_a{\tT}^a$ that takes values in $\fg^\star$, which gives
\be
\delta A=0\,,\qquad \delta \tilde A=\dr \tilde \Lambda +[A,\tilde \Lambda]\,,
\ee
where on the right-hand side of the last formula we recognise the covariant derivative with respect to $A$ of $\tilde A$, seen as a matter field in the coadjoint representation.

\subsubsection{The Third Way theory}
\label{sec:3rdway}
We now select $\bbD = G\times G$ (or, equally well, any $\mathbb D$ whose Lie algebra is a direct sum $\fg\oplus \fg$).  This will give us the Third Way theory, introduced in \cite{Arvanitakis:2015oga}. (We refer to the same reference for a more detailed discussion as well as to references \cite{Deger:2022znj,Deger:2021fvv}.) The Third Way theory describes a sector of ABJM theory (see e.g. \cite{Mukhi_2011,Nilsson_2014}), which has a $G\times G \to (G\times G)_{\diag}\cong G$ broken gauge symmetry, where by $(G\times G)_{\diag}$ we denote the diagonal subgroup. This is given a Manin pair structure as follows: We denote the generators of the Lie algebra $\fd = \fg\oplus\fg$ as $T_a^\pm$ satisfying
\begin{align}
    [T_a^\pm,T_b^\pm] &= f_{ab}{^c}T_c^\pm
    &
    \eta(T_a^\pm,T_b^\pm) &= \pm\kappa_{ab}
    \\
    [T_a^\pm,T_b^\mp] &= 0
    &
    \eta(T_a^\pm,T_b^\mp) &= 0
\end{align}
where $\kappa$ an invariant form on $\fg$. With this, the diagonal subalgebra $(\fg\oplus\fg)_{\diag} \cong \fg$ is isotropic (maximally so, in fact, for dimension-counting reasons). The Manin pair is thus given by $(\fg_\text{diag}\hookrightarrow\fg\oplus \fg,\eta)$. 

To write out the action explicitly we choose the antidiagonal subspace for $\tfg$, as it is also isotropic. Accordingly, $\fg_{\diag}$ and $\tfg$ have respective bases
\begin{align}
    T_a &= T_a^+ + T_a^-
    &
    \tT^a_k &= \frac{1}{2}k^{-1}\kappa^{ab}\big(T_b^+ - T_b^-\big)
\end{align}
satisfying relations
\begin{align}
    [T_a,T_b] &= f_{ab}{^c}T_c
    &
    \eta(T_a,T_b) &= 0
    \vphantom{\frac{1}{4}}\\
    [\tT^a_k,T_b] &= f_{bc}{^a}\tT^c
    &
    \eta(\tT^a_k,T_b) &= k^{-1}\delta^a{_b}
    \vphantom{\frac{1}{4}}\\
    [\tT^a_k,\tT^b_k] &= \frac{1}{4}k^{-2}f^{abc}T_c
    &
    \eta(\tT^a_k,\tT^b_k) &= 0
\end{align}
where indices on the $\fg$-structure constants $f$ are raised/lowered using $\kappa$. 

The reason for introducing $k$ into our basis is the observation that in the limit $k\to\infty$, the algebra of $G\times G$ reduces to that of $T^\star G$, which is mathematically analogous to how  Minkowski spacetime is the large radius limit of anti-de Sitter spacetime. This limit is in agreement with the observation that the Third Way theory is a continuous deformation of Yang-Mills theory, as taking $\bbD = T^\star G$ corresponds to the first order formulation of Yang-Mills theory we described in section \ref{Yang-Mills as a Manin theory}.

We now define the mass matrix $M$ by
\begin{align}
    MT_a &= 0\,,
    &
    MT^+_a &= +2kT_a\,,
    \\
    M\tT^a &= k\kappa^{ab}T_b\,,
    &
    MT^-_a &= -2kT_a\,.
\end{align}
Having specified the geometric data, we can now formulate the theory. To this end we expand the connection 1-form as 
\begin{equation}
    \bbA = A^{+a}T_a^+ + A^{-a}T_a^- = A^aT_a + \tA_a\tT^a_k
\end{equation}
where under $\fg_{\diag}\cong\fg$ we see that $A^\pm$ and $A$ are $\fg$-connections and $\tA$ is an auxiliary matter field in the adjoint $\fg$-representation. With this, action \ref{eq:ManinLag} becomes
\begin{equation}
\label{eq:thirdwayLag}
\begin{split}
    S[\bbA]
    &= kS_\CS[A^+] - kS_\CS[A^-] + kg^2\int\kappa_{ab}(A^+-A^-)^a\star(A^+-A^-)^b
    \\
    &= \int\bigg[\tA_aF^a + \frac{1}{24k^2}f^{abc}\tA_a\tA_b\tA_c + \frac{1}{2}g^2\kappa^{ab}\tA_a\star\tA_b\bigg].
\end{split}
\end{equation}
where $F = \extd A + A^2$. We see that in the $k\to\infty$ limit, this theory reduces to the first order formulation of Yang-Mills theory with coupling $g^2$, as formulated by action \ref{first order YM}.

A more direct way to see this is by solving for $\tA$ and back-substituting it into the action \cite{Mukhi_2011}. Its field equations are given and recursively solved by
\begin{gather}
    F + \frac{1}{4k^2}\tA^2 + g^2\star M\tA = 0
    \\
    \Rightarrow \tA_a = -\frac{1}{g^2}\star F_a - \frac{1}{8k^2g^6}f_{abc}\star({\star} F^b{\star} F^c) + \Or(k^{-4})
\end{gather}
Back-substitution into the action now yields
\begin{equation}
    S[\bbA] = \int\bigg[-\frac{1}{2g^2}\kappa_{ab}F^a\star F^b + \frac{1}{12k^2g^6}f_{abc}{\star} F^a{\star} F^b{\star} F^c + \Or(k^{-4})\bigg]
\end{equation}
which indeed expands the theory in powers of $k^{-2}$ around Yang-Mills theory.

Although we will not be employing this fact in this paper, the Third Way theory is also consistent when the Chern-Simons levels in \eqref{eq:thirdwayLag} are unequal, which may be achieved by adding Chern-Simons actions for either factor of $G$ or for $G\times G$ with a different bilinear form (which we call $\mathcal E$ much later in this paper).

\subsubsection{`Imaginary' Third Way}
\label{sec:ImaginaryThirdWay}
Now we move on to the complexified gauge group $\bbD=G_{\mathbb C}$. This group has a Lie algebra $\fg_{\mathbb C} = \fg \oplus i\fg$ and is ---as we will see--- similar to the Third Way case in many ways. Its algebra has generators $\{T_a\}$ for $\fg$ and $\{iT_a\}$ for $i\fg$  satisfying
\begin{align}
    [T_a,T_b] &= f_{ab}{^c}T_c
    &
    [T_a,iT_b] &= f_{ab}{^c}iT_c
    &
    [iT_a,iT_b] &= -f_{ab}{^c}T_c
\end{align}
Let us further introduce an invariant non-degenerate inner product $\eta = \kappa_{\mathbb C} - \overline{\kappa_{\mathbb C}}$,  the imaginary part of the complexification of an invariant form over $\fg$. This is
\begin{align}
    \eta(T_a,T_b) &= 0
    &
    \eta(T_a,iT_b) &= 2\kappa_{ab}
    &
    \eta(iT_a,iT_b) &= 0
\end{align}
It is now clear that we obtain a Manin pair $(\fg\hookrightarrow\fg_{\mathbb C},\eta)$. Let us now derive the Manin theory in close analogy to how we treated the Third Way case. We define a basis
\begin{equation}
    \tT^a_k = \frac{1}{2}k^{-1}\kappa^{ab}iT_b
\end{equation}
for $i\fg$, yielding
\begin{align}
    [T_a,T_b] &= f_{ab}{^c}T_c\,,
    &
    \eta(T_a,T_b) &= 0\,,
    \bigg.\\
    [\tT^a_k,T_b] &= f_{bc}{^a}\tT^c\,,
    &
    \eta(\tT^a_k,T_b) &= k^{-1}\delta^a{_b}\,,
    \bigg.\\
    \label{eq:minus sign}
    [\tT^a,\tT^b] &= {-}\frac{1}{4}k^{-2}f^{abc}T_c\,,
    &
    \eta(\tT^a,\tT^b) &= 0\,,
    \bigg.
\end{align}
which completely agrees with the $G\times G$ case, up to the minus sign in equation \ref{eq:minus sign}!

Upon introducing a gauge field $\bbA = A^aT_a + \tA_a\tT^a_k$ and define $M$ as
\begin{align}
    MT_a &= 0
    \\
    M\tT^a_k &= k\kappa^{ab}T_n
\end{align}
we arrive at the action
\begin{equation}
    S[\bbA] = \int\bigg[\tA_aF^a {-} \frac{1}{24k^2}f^{abc}\tA_a\tA_b\tA_c + \frac{1}{2}g^2\kappa^{ab}\tA_a\star\tA_b\bigg].
\end{equation}
Much like the Third Way theory, this reduces to Yang-Mills theory with a coupling $g^2$ in the $k\to\infty$ limit. 

\subsection{Equivalences of Manin theories from ``twists'' of Lie quasi-bialgebras}

We address the question: \emph{how many Lie quasibialgebra structures give rise to the same gauge theory?} This is relevant whenever we need or want to specify the theory by Lie quasibialgebra data instead of just the Manin pair/Dirac structure, which will be relevant for the Hamiltonian formulation of the theory.

Bangoura and Kosmann-Schwarzbach describe a \emph{twist} operation \cite{bangoura1993double} which takes a Lie quasibialgebra $(\fg,\tilde f,\tilde h)$ to another Lie quasibialgebra $(\fg,\tilde f',\tilde h')$ that fixes the underlying Lie algebra $\fg$. This is given in terms of a skew matrix $R^{ab}$ acting as
\be
\label{eq:explicittwist}
\begin{split}
    \tilde f^{ab}{}_c &\to \tilde f^{ab}{}_c +2 R^{d[a} f_{dc}{}^{b]}\\
    \tilde h^{abc}&\to \tilde h^{abc} + 3 \tilde f^{[ab}{}_d R^{c]d} -3 R^{[a|e|} R^{b|d|} f_{de}{}^{c]}
\end{split}
\ee
The skew matrix $R$ has a geometric interpretation: it parameterises deformations of a fixed lagrangian complement $\tfg$ onto a neighbouring one $\tfg'$. Explicitly, a basis of $\tfg'$ is given as $\tilde T^a{}'=\tilde T^a+ R^{ab}T_b$ and isotropicity forces $R$ to be skew.

 Twists define an equivalence relation on Lie quasibialgebras. We will show that \emph{equivalent Lie quasibialgebras yield equivalent field theories} at the level of explicit actions \eqref{eq:ManinLagExplicit}.

One way to see that this twist operation leads to equivalent gauge theories is via explicit calculation in the Hamiltonian formulation of the theory, that we will pursue later. An alternative more elegant argument is via the AKSZ construction \cite{Alexandrov:1995kv} for Chern-Simons theory. To that end, and to prove that twists take the above form in our conventions, we display an alternative ``BRST-esque'' formulation of Lie quasibialgebras.

\paragraph{A ``BRST'' version of the ``big bracket'' construction of \cite{bangoura1993double}.} Firstly recall the Chevalley-Eilenberg differential defining Lie algebra cohomology: introduce anticommuting variables $c^a$ with $abc$ $\frak g$-valued indices. Then the CE differential may be written
\be
Q\equiv \frac{1}{2} f_{bc}{}^a c^b c^c\frac{\pd}{\pd c^a}
\ee
and $Q^2=0$ is equivalent to the Jacobi identities for $f_{bc}{}^a$ the $\frak g$ structure constants. If the Lie algebra has an invariant nondegenerate symmetric bilinear form $\kappa$, then $Q$ is a hamiltonian vector field for the following Poisson structure
\be
\{c^a,c^b\}=\kappa^{ab}\implies Q=\{\tfrac{1}{2} f_{bc}{}^d\kappa_{da},\bullet\}\,,
\ee
and the Jacobian identity may also be written in the style of the `classical master equation' $\{\tfrac{1}{2} f_{bc}{}^d\kappa_{da}c^a c^b c^c,\tfrac{1}{2} f_{bc}{}^d\kappa_{da}c^a c^b c^c\}=0$. 

We may define various kinds of (quasi-)bialgebra structures on $\frak g$ by introducing alongside $c^a$ their duals $\tilde c_a$ along with the Poisson bracket $\{c^a,\tilde c_b\}=\delta^a_b\,,\{c,c\}=\{\tilde c,\tilde c\}=0$ and permutations thereof. (This Poisson bracket captures the inner product $\eta$.) In all cases there is a Lie algebra structure on $\frak{g}+\frak{g}^\star$ compatible with the split signature symmetric bilinear form defined by this Poisson bracket. A \emph{Lie quasibialgebra structure} on a Lie algebra $\frak g$ with structure constants $f_{bc}{}^a$ is equivalently defined via the quantities $f_{ab}{}^c,\tilde f^{bc}{}_a,\tilde h^{abc}$ appearing in the hamiltonian function
\be
\Theta=\tfrac{1}{2}f_{ab}{}^cc^ac^b \tilde c_c+\tfrac{1}{2}\tilde f^{bc}{}_a\tilde c_b\tilde c_c c^a + \tfrac{1}{3!} \tilde h^{abc} \tilde c_a \tilde c_b\tilde c_c\equiv f+\tilde f + \tilde h\,.
\ee
(If we were to switch on a term $c^3$ this would be a ``proto-Lie bialgebra''.) These quantities must satisfy, by definition
\be
\{\Theta,\Theta\}=0\,,
\ee
on top of the Lie algebra Jacobi identity for $\frak g$ ($\{f,f\}=0$), and by counting powers of $c,\tilde c$ we find the above condition is equivalent to
\be
Q_f\tilde f=0\,,\quad Q_{\tilde f}\tilde h=0\,,\quad \{\tilde f,\tilde f\}+2 Q_{f}\tilde h=0
\ee
where $Q_f\equiv \{f,\bullet\}=\{\tfrac{1}{2}f_{ab}{}^cc^ac^b,\bullet\}$ is the vector field with hamiltonian $f$. A Lie bialgebra is the case $\tilde h=0$. The Lie algebra over $\frak{g}+\frak{g}^\star$ defined by $\Theta$ is the  double $\frak d$ of $\fg$ constructed from the (quasi)-bialgebra structure.

By twist we mean a canonical transformation generated by a function of the form
\be
\Psi=\frac{1}{2} R^{ab}\tilde c_a\tilde c_b\,.
\ee
Here $R^{ab}$ is necessarily antisymmetric. This transforms the hamiltonian function to
\be
\Theta\to e^{\{\Psi,\bullet\}} \Theta= \Theta + \{\Psi,\Theta\} +\frac{1}{2}\{\Psi, \{\Psi,\Theta\}\}\,;
\ee
due to the form of $\Theta$ and $\Psi$ it is easy to confirm that the series terminates as indicated.

Since $e^{\{\Psi,\bullet\}}$ is a canonical transformation it is trivial to check that $\{e^{\{\Psi,\bullet\}}\Theta, e^{\{\Psi,\bullet\}}\Theta\}=0$ and also that the term $c^2\tilde c$ is invariant, so that $e^{\{\Psi,\bullet\}}\Theta$ defines a new Lie quasi-bialgebra. Explicit calculation leads to the formulas \eqref{eq:explicittwist} for the twist.

\paragraph{An AKSZ-based argument for twist invariance.} Let us pretend, in the first instance, that the mass term involving $M^{ab}$ in the Manin theory action \eqref{eq:ManinLagExplicit} is zero. Then we are dealing with pure Chern-Simons theory, which admits an AKSZ sigma model construction \cite{Alexandrov:1995kv} via the space of supermaps $\maps(T[1] M_3,\fg[1]\oplus \fg^\star[1])$\footnote{For the Lie bialgebra case see \cite{Arvanitakis:2021lwo} for more details, to which we also refer for QP-manifold conventions.}. The target $\fg[1]\oplus \fg^\star[1]$ is the purely fermionic supermanifold with coordinates $c^a,\tilde c_a$ that we were using above.

The fields of the AKSZ sigma model are, as usual, the component fields of the superfields $c^a(\sigma,\dr \sigma)$ and $\tilde c_a(\sigma,\dr \sigma)$ with $\sigma$ coordinates on $M_3$ and $\dr \sigma$ their differentials. The 1-form component fields may be identified as the 1-forms $A^a,\tilde A_a$ appearing in \eqref{eq:ManinLagExplicit}. Canonical transformations on the AKSZ target $\fg[1]\oplus \fg^\star[1]$ lift to canonical transformations for the graded symplectic structure on $\maps(\cdots)$, namely (the inverse of) the BV antibracket. Therefore, the lift of the canonical transformation $e^{\{\Psi,\bullet\}}$ displayed above is necessarily an equivalence of BV theories.

On the fields of the AKSZ sigma model this transformation acts by
\be
c^a(\sigma,\dr \sigma)\to c^a(\sigma,\dr \sigma) + R^{ba}\tilde c_b(\sigma,\dr \sigma)\,,\qquad \tilde c_b(\sigma,\dr \sigma)\to \tilde c_b(\sigma,\dr \sigma)
\ee
so in particular it acts as $A^a \to A^a + R^{ba}\tilde A_b$ and $\tilde A\to \tilde A$ on the original gauge fields appearing in the action \eqref{eq:ManinLagExplicit}.

The introduction of the mass term breaks the master equation obtained via the AKSZ construction, which is anyway expected. However, assuming that the master equation for the BV formulation of the Manin theory can be written down using the fields of $\maps(\cdots)$, we observe that upon switching off all of the ghosts and antifields, the canonical transformation leaves the mass term invariant, since the latter only depends on $\tilde A$ which does not transform.

The moral here is that the twisting of Lie quasi-bialgebras may be undone by the field redefinitions $A^a \to A^a + R^{ba}\tilde A_b$ and $\tilde A\to \tilde A$ applied to the action \eqref{eq:ManinLagExplicit}.

\subsection{Hamiltonian formulation; classification in the quasi-triangular case}
\label{sec:hamiltonian}
We now construct the Hamiltonian action for the theory. This goes through identical steps to the calculation of the Hamiltonian for the \trd{} theory \cite{Arvanitakis:2015oga}. As such we will be brief. There are two outcomes of this analysis: first, that (in Lorentzian signature) the Hamiltonian is \emph{bounded below}; second, that for the class of \emph{quasi-triangular Lie quasibialgebras} --- which includes the perhaps more familiar quasi-triangular Lie bialgebras, defined via a Yang-Baxter equation --- Manin theories may be completely classified.

Since the time derivatives appear as in Chern-Simons theory, the off-shell phase space is, a priori, unaffected by the mass term in the action \eqref{eq:ManinLag}. Therefore we employ the usual split $\mathbb A\equiv \mathbb A^0\dr t +\bbalpha$, where we assume spacetime is (locally) of the form $M_3=\mathbb R\times \Sigma$, $\mathbb R$ being ``time'' $t$, and $\Sigma$ being 2-dimensional ``space''. In pure Chern-Simons theory, $\mathbb A^0$ become lagrange multipliers enforcing the vanishing of the Gauss law constraint $\bbchi$:
\be
\bbchi\equiv \dr_\Sigma \bbalpha +\bbalpha^2\,.
\ee
(Above $\dr_\Sigma$ is the de Rham differential along $\Sigma$ and $\bbalpha$ is a $\fd$-valued 1-form on $\Sigma$ for each value of $t$.)

For Manin theory, the difference arises from the mass term. What happens is that half of the $\mathbb A^0$ is no longer a lagrange multiplier:  assuming $g(\dr t,\dr t)\equiv g^{tt}\neq 0$, we have $\bbalpha \star \dr t=\dr t \star \bbalpha=0$ whence
\be
\tfrac{1}{2}\eta(\mathbb A \star M\mathbb A)=\tfrac{1}{2}\eta( \mathbb A^0 M\mathbb A^0 (\dr t\star \dr t)+ \bbalpha \star M\bbalpha)\,.
\ee
With the same assumption, $\dr t\star \dr t$ is everywhere nonvanishing and proportional to the volume form on $M_3$. Given $M\mathbb A^0=M\tilde A^0$  (due to \eqref{eq:Midentities}) we see we may integrate out $\tilde A^0$, replacing it with part of the Gauss law constraint for pure Chern-Simons theory:
\be
\chi^a\equiv \dr_\Sigma \alpha^a +\frac{1}{2} f_{bc}{}^a \alpha^b\alpha^c + \tilde f^{ca}{}_b \alpha^b \tilde \alpha_c + \frac{1}{2}\tilde \alpha_b\tilde\alpha_c \tilde h^{bca}\,.\qquad (\bbchi\equiv \chi^a T_a + \tilde \chi_a \tilde T^a\,.)
\ee
Note that we have split $\fd$-valued objects into $\fg$ and $\tfg$ parts in the usual way $(\bbchi\equiv \chi^a T_a + \tilde \chi_a \tilde T^a\,, \bbalpha=\alpha^a T_a + \tilde \alpha_a \tilde T^a\,, \mathbb A^0\equiv A^0{}^a T_a + \tilde A^0_a\tilde T^a\,.)$ To write the Hamiltonian we need to trade $\chi^a$ (a 2-form on $\Sigma$) for a scalar $\check\chi^a$, which to be concrete we define by $\chi^a\equiv \check\chi^a \star \dr t$ (where $\star$ remains the 3-dimensional Hodge star).

After eliminating $\tilde A^0$ we obtain the Hamiltonian action, equivalent to \eqref{eq:ManinLag}:
\be
S=\int_{M_3} \dr t \big(\eta(-\tfrac{1}{2}\bbalpha \dot \bbalpha) +   A^0{}^a\tilde\chi_a\big) +\tfrac{1}{2}\big(\tilde \alpha_a{}_\mu \tilde \alpha_{b}{}_\nu M^{ab} g^{\mu\nu}-g^{tt}\check\chi^a\check\chi^b M_{ab}\big) \mathrm{vol}_g\,,
\ee
with $\mathrm{vol}_g\equiv {\star}1$ being the volume form associated to the metric $g$ on $M_3=\mathbb R\times  \Sigma$ and
\be
\tilde\chi_a= \dr_\Sigma \tilde \alpha_a +\tfrac{1}{2}\tilde f^{bc}{}_a \tilde \alpha_b\tilde\alpha_c - \alpha^b \tilde\alpha_c f_{ba}{}^c
\ee
coming from the Gauss law constraint $\bbchi$ for Chern-Simons theory. It is trivial to check that $\tilde \chi_a$ generates the gauge transformations \eqref{eq:infinitesimalgaugetransfs} and thus the proof of gauge invariance in that subsection shows that the Hamiltonian along with these constraints form a first-class constrained hamiltonian system.

\paragraph{Boundedness of the Hamiltonian.} We read off the Hamiltonian function as
\be
\label{eq:Maninhamiltonian}
H\equiv \frac{1}{2}\Big(\tilde \alpha_a{}_\mu \tilde \alpha_{b}{}_\nu M^{ab} g^{\mu\nu}-g^{tt}\check\chi^a\check\chi^b M_{ab}\Big)\,.
\ee
This is non-negative when
\begin{enumerate}
    \item the matrices $M^{ab}$ and $M_{ab}$ are positive-definite, which may always be arranged when $\fg$ is a compact Lie algebra, and
    \item $M_3$ is a Lorentzian-signature spacetime,  we have chosen coordinates where $g_{\mu\nu}$ is block-diagonal between time and space\footnote{This is indeed a choice of gauge for the background metric, as may be seen via a lapse function/shift vector parameterisation of $g$, ADM-style.}, and $t$ is a timelike coordinate.
\end{enumerate}
Then the second term is a sum of positive squares. Moreover since $\tilde \alpha_{a0}=0$ by definition, the first term is also a sum of squares by block-diagonality.

\paragraph{Classification of Manin theories in the quasi-triangular case.} Quasi-triangular Lie quasibialgebras are defined by Bangoura and Kosmann-Schwarzbach \cite[Definition 1.7]{bangoura1993double}, in terms of a skew matrix $R^{ab}$, an $\ad\fg$-invariant symmetric form $s^{ab}$, and the 3-form $\tilde h^{abc}$, that solve a classical Yang-Baxter-type equation. In our conventions we may derive the equation via the ``BRST'' approach of the previous subsection: assuming $\tilde f$ is \emph{coboundary}\footnote{These are called \emph{exact} in \cite{bangoura1993double}.}\,,
\be
\label{eq:coboundary}
\tilde f^{ab}{}_c=2 R^{d[a} f_{dc}{}^{b]}\,, \iff \tilde f=-Q_f\underbrace{(\tfrac{1}{2} R^{ab}\tilde c_a\tilde c_b)}_{\equiv R}
\ee
for some skew matrix $R^{ab}$, we have a Lie quasibialgebra if and only if $Q_f(\{R,Q_f R\} +2\tilde h)=0$. A \emph{triangular} one is defined to be such that $\{R,Q_f R\} +2\tilde h=0$; for a \emph{quasi-triangular} one the right-hand side is allowed, by definition, to be proportional to $f_{de}{}^a s^{bd}s^{ec} \tilde c_a\tilde c_b\tilde c_c$, which is annihilated by $Q_f$ when $s^{ab}$ is $\operatorname{ad}\fg$-invariant.

To write down a concrete and index-free expression for the aforesaid equation we invoke some nondegenerate inner product $\kappa_{ab}$ on $\fg$ to define operators $R:\fg\to\fg$ from the skew matrix $R^{ab}$ and $s:\fg\to\fg$ from the symmetric matrix $s^{ab}$ via $R(T_a)\equiv R_a{}^b T_b$ and via $s(T_a)= s_a{}^b T_b$ respectively; then this data must solve
\be
\label{eq:modifiedmodifiedCYBE}
[Rx,Ry]-R([Rx,y]+[x,Ry])+[sx,sy]=\tilde h(x,y)\,,\qquad \forall x,y\in\fg\,.
\ee
(We also defined $\tilde h(x,y)= \tilde h_{ab}{}^cx^a y^b T_c$.) When $\tilde h=0$ and $s^{ab}=c\kappa^{ab}$ for some number $c$, we recognise the usual modified Classical Yang-Baxter Equation (mCYBE), solutions of which define Lie bialgebras\footnote{For an exposition of these in a physics context we refer to a work by Vicedo \cite{Vicedo:2015pna}.}. For this reason we shall call equation \eqref{eq:modifiedmodifiedCYBE} the \emph{amended} mCYBE, or \emph{amCYBE}.


We will show that the Manin theory depends, in this case, only on $s$ and $\fg$ via the 3-form constructed above, or, equivalently, via the map
\be
x,y\to[sx,sy]\,.
\ee For this, it is convenient to write $\fd$-valued objects not in terms of the basis spanned by $T_a$ and $\tilde T^a$ but in terms of $T_a$ and $\tilde T_a=\kappa_{ab}\tilde T^b$, so that the double $\frak d$ is a sum of two copies of $\frak g$  \emph{as a vector space}, and $\eta(T_a,\tilde T_b)=\kappa_{ab}$ are the nonvanishing matrix elements up to permutations. We will also rewrite $\tilde \alpha$ (valued in $\tfg)$ as the $\fg$-valued form
\be
\Pi\equiv \Pi^a T_a=\tilde \alpha_b\kappa^{ba}T_a\,.
\ee
In this notation the phase space action is (up to a temporal integration by parts)
\be
\begin{split}
 S=\int_{\mathbb R\times \Sigma}\dr t\Big(\dot \alpha^a \kappa_{ab}\Pi^b + A^{0a}\underbrace{\big(\dr_\Sigma \Pi_a +\tfrac{1}{2}\tilde f^{bc}{}_a \Pi_b\Pi_c -\alpha^b \Pi_c f_{ba}{}^c\big)}_{=\tilde \chi_a}\Big)  +\mathrm{vol}_g H\\
 \chi^a=\dr_\Sigma \alpha^a +\frac{[\alpha,\alpha]^a}{2}+\alpha^b \Pi_c \tilde f^{ca}{}_b+\frac{1}{2}\tilde h^{abc}\Pi_b\Pi_c\,.
\end{split}
\ee
We now recognise $\Pi$ as the conjugate momentum corresponding to $\alpha$.

So far this is completely general. Assuming that $\frak g$ is quasi-triangular, we may simplify by removing all mentions of $\tilde f$ via a canonical transformation:
\be
\alpha\to \alpha -R(\Pi)\,, \Pi\to \Pi.
\ee

\begin{claim}Under said transformation,
\be
\begin{split}
\tilde \chi_b \kappa^{ba}T_a\to \dr_\Sigma \Pi+[\alpha,\Pi]\,,\\
\chi^aT_a\to \dr \alpha +\tfrac{1}{2}[\alpha,\alpha]- R(\dr_\Sigma \Pi+[\alpha,\Pi])+\frac{1}{2} [s\Pi,s\Pi]\,,
\end{split}
\ee
where all of the brackets are now ones in $\frak g$ and not $\frak d$.
\end{claim}
\begin{proof}
Rewrite
\be
\tilde \chi=\tilde \chi_a \tilde T^a
\ee
as the following closely related $\frak g$-valued 2-form
\be
\tilde \chi^\kappa\equiv \tilde \chi_a \kappa^{ab}T_b
\ee
which immediately gives (using the fact $\tilde f$ is coboundary)
\be
\tilde\chi^\kappa=\dr_\Sigma \Pi +[\alpha,\Pi]+\tfrac{1}{2} \tilde f^{bc}{}_d \Pi_b\Pi_c\kappa^{ad}T_a=\dr_\Sigma \Pi +[\alpha,\Pi]+\tfrac{1}{2}[\Pi,\Pi]_R=\dr_\Sigma \Pi +[\alpha+ R(\Pi),\Pi]\,.
\ee
which goes to $\dr_\Sigma \Pi +[\alpha,\Pi]$ under the canonical transformation.

It remains to calculate $\chi$. This contains the following annoying term which we write in terms of $R$ and the original bracket $[,]$ on $\frak g$ again:
\be
\alpha^b\Pi_c\tilde f^{ca}{}_bT_a=[\alpha,R\Pi]-R[\alpha,\Pi]\,.
\ee
This now makes the calculation for $\chi$ straightforward:
\be
\chi\to (\dr_\Sigma \alpha + \alpha^2)-R(\dr_\Sigma \alpha + [\alpha,\Pi])+R[R\Pi,\Pi]-\tfrac{1}{2}[R\Pi,R\Pi] +\tfrac{1}{2}\tilde h(\Pi,\Pi)\,.
\ee
The Claim then follows from the amCYBE \eqref{eq:modifiedmodifiedCYBE}.\end{proof}

Note that we may now remove the term involving $R$ from the action because it is proportional to the constraint enforced by $A^{0a}$. (E.g.~by redefining this multiplier.) Having done this, the phase space action becomes
\be
\begin{split}
 S=\int_{\mathbb R\times \Sigma}\dr t\;\kappa\big(\dot \alpha\Pi + A^0(\dr_\Sigma \Pi+[\alpha,\Pi])\big) + \mathrm{vol}_g H \,.
\end{split}
\ee
where the Hamiltonian $H$ takes the original form \eqref{eq:Maninhamiltonian} except for the expression for $\chi$:
\be
H=\tfrac{1 }{2}\big(\Pi_{ia}\Pi_{jb} M^{ab}g^{ij}- g^{00} \check\chi^a\check\chi^b M_{ab}\big)\,,\qquad \chi=\dr_\Sigma \alpha +\tfrac{1}{2}[\alpha,\alpha]+\tfrac{1}{2} [s\Pi,s\Pi]\,,
\ee
and where, again, $\check \chi^a$ is the scalar dual to the 2-form $\chi^a$ along $\Sigma$ ($\chi^a=\check \chi^a {\star} \dr t$).

This is the classification result. We note that the same result could be obtained more abstractly via twists: it is proven in \cite[Proposition 1.8]{bangoura1993double} that any quasi-triangular Lie quasibialgebra is equivalent to one with $\tilde f=0$ and $\tilde h$ equalling the degree 3 cocycle $f_{de}{}^a s^{bd}s^{ec} \tilde c_a\tilde c_b\tilde c_c$ that contributes to the  amCYBE \eqref{eq:modifiedmodifiedCYBE}. 

\subsubsection{Example: theories of mCYBE type and their `duals'} In this case by definition $\tilde h=0$ and $s^{ab}=c\kappa^{ab}$, $\kappa$ being the inverse of a nondegenerate inner product on $\fg$, so that the amCYBE \eqref{eq:modifiedmodifiedCYBE} becomes the usual mCYBE:
\be
\label{eq:mCYBE}
[Rx,Ry]-R([Rx,y]+[x,Ry])+c^2[x,y]=0\,.
\ee

Solutions of the mCYBE for a real Lie algebra $\fg$ fall into three cases depending on the sign of $c^2$. This sign determines the structure of the double $\fd$ of $\fg$ as follows:
\be
\begin{cases} c^2>0: &\fd \cong \fg\oplus \fg\,,\qquad \text{``factorisable bialgebra''}\\ c^2=0: &\fd \cong \fg\ltimes \fg^\star\,,\qquad\text{``triangular bialgebra''} \\ c^2<0: &\fd \cong \fg_{\mathbb C}\,,\qquad\text{``imaginary bialgebra''}
\end{cases}
\ee
The terminology here is not entirely standard; in particular we do not follow Vicedo \cite{Vicedo:2015pna} where a compact proof of the classification may be found. The associated Manin pairs $(\fg\into\fd)$ and Manin theories are
\begin{itemize}
    \item \underline{factorisable case:} $\fg$ embeds as the diagonal subalgebra in the direct sum $\fd=\fg\oplus \fg$ where the summands commute. The associated Manin theory is the \trd{} theory described in Sec.~\ref{sec:3rdway}.
    \item \underline{triangular case:} $\fg$ is as indicated in the semidirect product $\fg\ltimes \fg^\star$, and the associated Manin theory is Yang-Mills theory (see Sec.~\ref{sec:YM}).
    \item \underline{imaginary case:} $\fg$ embeds as the indicated real form of the complexification $\fg_{\mathbb C}$. The Manin theory is described in Sec.~\ref{sec:ImaginaryThirdWay}.
\end{itemize}

\paragraph{`Duals' to mCYBE type theories.} Solutions to the mCYBE \eqref{eq:mCYBE} define Lie bialgebras, not just Lie quasibialgebras. This means that $\tfg\cong \fg^\star$ is a Lie algebra. If we use $\kappa_{ab}$ to identify $\fg\cong\tfg$ we find a second Lie algebra structure on $\fg$ given by
\be
\label{eq:Rbracket}
[x,y]_R\equiv [Rx,y]+ [x,Ry]\,.
\ee
(This is the coboundary equation \eqref{eq:coboundary} up to index gymnastics.)

In these cases we therefore have a Manin triple structure $(\fd,\fg,\tfg)$ instead of just a Manin pair. In particular $(\tfg\into \fd)$ is a Manin pair --- in some sense dual to the original pair $(\fg\into \fd)$ --- and thus every solution to the mCYBE gives rise to a \emph{pair} of Manin theories. We have already seen an example of a `dual' Manin theory: for $R=0$ we have a triangular solution to the mCYBE where $\tfg$ is abelian; this gives rise to the Freedman-Townsend theory of Sec.~\ref{sec:freedmantownsend}. We emphasise however that these `duals' are not unique.

As an example of this phenomenon, we exhibit a different triangular structure. In general, triangular bialgebra structures on a Lie algebra $\fg$ correspond to subalgebras with an invariant nondegenerate 2-form on them; in that case such $\fg$ are called quasi-Frobenius\footnote{See again Vicedo \cite[section 2.3]{Vicedo:2015pna} for a quick proof.}. Triangular structures are classified explicitly  for $\mathfrak{sl}(3,\mathbb C)$ in the Chari-Pressley book \cite[Example 3.1.8]{Chari:1994pz}. Of those, example (v)(a) descends to the compact real form $\fg=\mathfrak{su}(3)$. The subalgebra in question is that of diagonal matrices in $\mathfrak{su}(3)$. This is the 2-dimensional abelian subalgebra spanned by
\be
H_1=i\text{diag}(1,-1,0)\,,\qquad H_2=i\sqrt{3}/3\text{diag}(1,1,-2)
\ee
where $\kappa(x,y)=\tr(xy)$. If we pick $\omega(H_1,H_2)=-\omega(H_2,H_1)=1$ for the 2-form, we calculate the corresponding $R$ as
\be
R(\text{non-diagonal})=0\,,\qquad R(H_1)=2H_2\,,\qquad R(H_2)=-2 H_1\,.
\ee
It is easy to verify this solves the mCYBE \eqref{eq:mCYBE} (with $c^2=0$). This example furnishes a `dual' to $\fg=\mathfrak{su}(2)$ Yang-Mills in its Manin theory formulation which has nonabelian gauge algebra $\tfg$ as given by \eqref{eq:Rbracket}.

It is tempting to conjecture that these `dual' Manin theories are honestly dual to each other in the physical sense. We leave this question for the future. 

\section{Evanescent supersymmetry}
\paragraph{Generalities.} Say we are studying a theory with action $S$, which admits a supersymmetric(-esque) generalisation to a theory $S_\text{SUSY}$. (Parenthetical to be explained shortly.) We will say $S_\text{SUSY}$ has \textbf{evanescent supersymmetry} when $S_\text{SUSY}$ and $S$ are equivalent actions, in the sense that integrating out fields from $S_\text{SUSY}$ yields $S$. (One way to formalise this is that $S$ and $S_\text{SUSY}$ define homotopy-equivalent $L_\infty$-algebras in the sense of \cite{Arvanitakis:2020rrk,Arvanitakis:2021ecw}.)

An immediate corollary of this definition is that $S_\text{SUSY}$ describes exactly the same number of on-shell degrees of freedom as $S$ does; in other words, $S_\text{SUSY}$ does not assign superpartners to the degrees of freedom described by $S$. The point of an evanescent-supersymmetric formulation is that it lends itself well to exact quantum calculations via localisation arguments.

We will see that for Manin gauge theory the symmetry transformations are partially deformed relative to conventional $N=2$ SUSY in 3D; in particular, the algebra is deformed by terms proportional to the Manin mass operator $M$. This is the ``price to pay'' for introducing SUSY in a theory without introducing superpartners to the on-shell degrees of freedom, and is why we expect in general the evanescent supersymmetry transformations to be ``supersymmetry-esque''.

\subsection{Duistermaat-Heckman formula as evanescent localisation}

Before moving to discuss path integrals, we first take a look at the familiar example of the Duistermaat-Heckman formula and interpret it from the viewpoint of evanescent supersymmetry. This subsection closely follows Pestun and Zabzine's exposition \cite{pestun2017introduction}.

We consider a compact symplectic manifold $(\gM,\omega)$ of dimension $2\ell$ and on this manifold a Morse function $H \in C^\infty(\gM)$ (which we think of as an `action') such that the orbits generated by a Hamiltonian vector field $K$ (characterised by $\iota_K\omega = \extd H$) are compact (a circle or a point). We then consider the `partition function' given by the oscillatory integral $Z_H$ which by the Duistermaat-Heckman formula is given by
\begin{equation}
\label{Duistermaat-Heckman 1}
    Z_H = \int_\gM \frac{\omega^\ell}{\ell!}e^{iH} = (2\pi i)^\ell\sum_{\extd H(x_0) = 0}\frac{e^{iH(x_0)}}{\sqrt{\det\del_\mu K^\nu}}\bigg|_{x = x_0}.
\end{equation}
Now, what we want to do in the context of evanescent supersymmetry is to view this as a partition function where we have integrated out auxiliary fermionic degrees of freedom. We reintroduce these simply by viewing it through the lens of integration on supermanifolds:
\begin{equation}
    Z_H 
    = \int_\gM \big[\extd^{2\ell}x\big]~\Pf\omega~ e^{iH}
    = (-i)^\ell\int_{\Pi T\gM} \big[\extd^{2\ell}x\big|\extd^{2\ell}\psi\big]~e^{iS(x,\psi)}
\end{equation}
where we integrate over the parity-reversed tangent bundle $\Pi T\gM$ with coordinates $(x^\mu,\psi^\mu)$ ($\psi^\mu$ being the odd fibre coordinates) a `supersymmetrised' action given by 
\begin{equation}
\label{eq:SactionDH}
    S(x,\psi) = H(x) + \frac{1}{2}\omega_{\mu\nu}(x)\psi^\mu\psi^\nu.
\end{equation}
It is supersymmetric in the sense that it is left invariant by the odd vector field
\begin{equation}
    Q = \psi^\mu\pder{}{x^\mu} - K^\mu\pder{}{\psi^\mu}
\end{equation}
i.e. $Qx^\mu = \psi^\mu$ and $Q\psi^\mu = -K^\mu$. We take $Q$ to act as an odd left-derivative.

The localisation procedure can then be reformulated as follows: We pick an arbitrary metric $g_{\mu\nu}$ of which $K^\mu$ is a Killing vector (this is always possible since the orbits of $K$ are assumed compact) and note that we can deform the partition function as
\begin{equation}
    Z_H(t) = \int_{\Pi T\gM} \big[\extd^{2\ell}x\big|\extd^{2\ell}\psi\big]~e^{iS + itS_\loc}
\end{equation}
by a $Q$-exact localisation term $S_\loc$ given by
\begin{align}
    S_\loc &= QF_\loc = K^\mu K_\mu - \del_\mu K_\nu\psi^\mu\psi^\nu,
    &
    F_\loc &= g_{\mu\nu}K^\mu\psi^\nu.
\end{align}
without changing it. Indeed, one finds that
\begin{equation}
    \pder{}{t}Z_H(t) 
    = 
    \int_{\Pi T\gM} \big[\extd^{2\ell}x\big|\extd^{2\ell}\psi\big]~iS_\loc e^{iS + itS_\loc} 
    = 
    \int_{\Pi T\gM} \big[\extd^{2\ell}x\big|\extd^{2\ell}\psi\big]~Q\bigg\{iF_\loc e^{iS + itS_\loc}\bigg\}
    =
    0
\end{equation}
where we used the fact that $\int_{\Pi T\gM}[\extd^{2\ell}x|\extd^{2\ell}\psi]Q(\bullet) = 0$\footnote{There exist pure supergeometric proofs of this invariance of the measure under odd vector fields \cite{Schwarz:1995dg}. Here it is easily seen by rewriting $Q$ in terms of the de Rham derivative  and contraction-with-$K$  operators.}.  Taking the limit $t\to\infty$ we then find that the integral localises to
\begin{align}
    S_\loc|_\bos &= 0
    &
    &\Leftrightarrow
    &
    K &= 0
    &
    &\Leftrightarrow
    &
    \extd H &= 0
\end{align}
which is to say, the same locus as what the Duistermaat-Heckman formula states.

In the usual fashion of localisation arguments we then introduce coordinates around saddle points $x_0$ of $H$ as
\begin{align}
&
\begin{aligned}
    x^\mu = x_0^\mu + \frac{y^\mu}{\sqrt{t}}\,,\quad
    \psi^\mu = \frac{\chi^\mu}{\sqrt{t}}
\end{aligned}
&
&\Rightarrow
&
\big[\extd^{2\ell}x\big|\extd^{2\ell}\psi\big] &= \big[\extd^{2\ell}y\big|\extd^{2\ell}\chi\big].
\end{align}
With this we can expand
\begin{equation}
    S(x,\psi) + tS_\loc(x,\psi) 
    = 
    H(x_0) + \frac{1}{2}
    (y,\chi)
    \big[\Hess S_\loc(x_0,0)\big]
    (y,\chi)^\top
    + \Or(t^{-1/2})
\end{equation}
where by $\Hess$ we denote the Hessian matrix of a given function. With this we find that the partition function is given by
\begin{equation}
    Z_H = (2\pi i)^\ell\sum_{S_\loc|_\bos(x_0) = 0}\frac{e^{iS}}{\sqrt{\Ber\Hess S_\loc}}\bigg|_{(x,\psi) = (x_0,0)}
\end{equation}
Upon evaluating the Hessian one finds that this reduces to the Duistermaat-Heckman formula. This just formulates it in a way which is more easily related to localisation in supersymmetric quantum field theory. To relate it to the Duistermaat-Heckman formula given in equation \ref{Duistermaat-Heckman 1} we expand
\begin{equation}
\label{Hessian}
    (y,\chi)
    \big[\Hess S_\loc(x_0,0)\big]
    (y,\chi)^\top
    = B_{\mu\nu}y^\mu y^\nu + F_{\mu\nu}\chi^\mu\chi^\nu
\end{equation}
and by invariance of \ref{Hessian} under the linearised supersymmetry
\begin{equation}
    Q_{\text{lin}} = \psi^\mu\pder{}{y^\mu} + \del_\mu K^\nu(x_0)y^\mu\pder{}{\chi^\nu}
\end{equation}
(with $\del_\mu = \del/\del x^\mu$) we find that
\begin{equation}
    B_{\mu\nu} = F_{\mu\rho}\del_\nu K^\rho(x_0)
\end{equation}
yielding
\begin{equation}
    \Ber\Hess S_\loc(x_0) = \frac{\det B_{\mu\nu}}{\det F_{\mu\nu}} = \det \del_\mu K^\nu(x_0).
\end{equation}
Using this we arrive at the original Duistermaat-Heckman formula \ref{Duistermaat-Heckman 1}.

We revisited this well-known result in order to, on the one hand, remind the reader of some basics of localisation, and on the other hand as a prototype of evanescent supersymmetry. We started with a purely bosonic theory which we rewrote using auxiliary fermionic degrees of freedom: indeed in \eqref{eq:SactionDH} the $\psi^\mu$ may be integrated out since $\omega_{\mu\nu}$ is invertible.
This in turn allowed us to make us of the methods of localisation to solve these integrals. This is exactly what we will do in the context of the quantum field theory: starting off with a partition function which only includes bosonic field content, by adding auxiliary degrees of freedom we allow ourselves to use the method of localisation through evanescent supersymmetry.

\subsection{Chern-Simons as an evanescent SUSY gauge theory}
Supersymmmetric Chern-Simons theory was originally written down in \cite{Lee:1990it}. We write down the action for $N=2$ Euclidean supersymmetry on flat space  $M_3=\mathbb R^3$ and arbitrary gauge algebra $\fg$ with invariant inner product $\kappa$. The fields are $A,\sigma,D,\lambda,\hat\lambda$, of which $A$ is a $\fg$-valued 1-form and $\sigma,D$ are $\fg$-valued scalars, and $\lambda,\hat\lambda$ are $\fg$-valued complex 2-component anticommuting spinor fields. The action  reads, in Cartesian coordinates,
\begin{equation}
    S^{N = 2}_\CS[A,\sigma,\lambda,\hlambda,D] 
    = 
   \int \kappa\Big(\frac{1}{2} A\dr A +\frac{1}{3}A^3 +\star(\lambda\hlambda - \sigma D)\Big)
\end{equation}
and we are using Northwest-Southeast conventions for the spinors so $\lambda\hat\lambda=\lambda^\alpha \hat\lambda_\alpha$.

This lagrangian is supersymmetric under the usual $N=2$ Euclidean supersymmetry transformations, where $\zeta$ and $\hat \zeta$ are 2-component complex bosonic spinors:
\begin{subequations}
\begin{align}
    \delta A_\mu &= -\zeta\gamma_\mu\hlambda + \lambda\gamma_\mu\hzeta
    \Big.\\
    \delta\sigma &= -i\zeta\hlambda + i\lambda\hzeta
    \Big.\\
    \delta\lambda &= i\zeta\Big(D + \sF - i\gsD\sigma\Big)
    \Big.\\
    \delta\hlambda &= \phantom{\zeta}i\Big(D + \sF + i\gsD\sigma\Big)\hzeta
    \Big.\\
    \delta D &= \gD_\mu\Big(\zeta\gamma^\mu\hlambda + \lambda\gamma^\mu\hzeta\Big) + i\Big[\sigma,\zeta\hlambda + \lambda\hzeta\Big]
    \Big.
\end{align}
\end{subequations}
with $\Dg$ the covariant derivative w.r.t.\ $G$. Since the supersymmetry multiplet is given in Wess-Zumino gauge, the algebra closes on translations along with compensating gauge transformations: for any field $\varphi$,
\be
\begin{split}
\label{eq:SUSYALGEBRAMASSLESS}
    \{\delta_\zeta,\hat\delta_{\hat\zeta}
\}\varphi=2i\mathcal L_K \varphi -2i\delta^\text{gauge}_\Lambda\varphi\,,\qquad K^\mu\equiv \zeta\gamma^\mu \hat \zeta\,, \quad \Lambda\equiv K^\nu A_\nu +i\zeta\hat\zeta \sigma\,,
\end{split}\ee
and the anticommutators of two hatted or two unhatted supersymmetries vanish:
\be
\{\delta_\zeta,\delta_{\zeta'}\}=\{\hat\delta_{\hat \zeta},\hat\delta_{\hat\zeta'}\}=0\,.
\ee
We display the full algebra in appendix.

We emphasise two points that will be relevant for supersymmetric Manin theory:
\begin{enumerate}
    \item We have $\delta_\zeta^2=\hat\delta_{\hat \zeta}^2=0$. This enables supersymmetric localisation.
    \item This $N=2$ theory is \emph{equivalent} to $N=2$ Chern-Simons theory: $\sigma,D,\lambda,$ and $\hat\lambda$ may be integrated out, leaving only $S_\CS[A]$. Therefore the supersymmetry is \emph{evanescent}. Note, however, that this theory is an exceptional case where evanescent supersymmetry is identical to conventional supersymmetry! The reason is that the field equation is $F=0$, which makes $\delta \lambda$ and $\delta \hat \lambda$ vanish on-shell.
    \item An important corollary of this last point is that the anticommutators of supersymmetry transformations must vanish on-shell, up to gauge transformations, so that the corresponding bosonic symmetries are trivial on $A$. This is indeed the case: for example, anticommutators on $A$ involve its field strength $F$ which vanishes on-shell.
\end{enumerate}
Irrespective of the fact the symmetry is trivial in the above sense, this supersymmetry yielded well-known exact results for certain observables via supersymmetric localisation \cite{kapustin2010exact}, and moreover the equivalence between $N=2$ and $N=0$ Chern-Simons theory has been confirmed exactly, at least for supersymmetric observables \cite{Fan_2019}.

\subsection{$N=2$ evanescent-supersymmetric Manin gauge theory}
Without further ado we display the action for $N=2$ evanescent-supersymmetric Manin theory:
\be
\label{eq:SUSYManinLagrangian}
\begin{split}
S_{N = 2} =\\ \frac{1}{2}\int_{M_3}\eta\bigg(k\Big(\bbA\dr\bbA +\tfrac{2}{3}\bbA^3\Big) +   g^2 \bbA\star M\bbA\bigg)+ k\int_{M_3}\star\eta\bigg(\bblambda\hat{\bblambda}-\bbsigma\mathbb D\bigg) + \frac{1}{2}\int_{M_3}\eta\Big(g^2\bbsigma\star M\bbsigma\Big)
\end{split}
\ee
This action may be seen as a mass deformation of the $N=2$ Chern-Simons theory described above, where now the fields $\mathbb A,\bbsigma,\mathbb D,\bblambda,\hat\bblambda$ take values in $\frak d$.

One may convince oneself that the supersymmetry $\delta_\text{CS}$  of $N=2$ Chern-Simons theory may be deformed into an invariance of the SUSY Manin theory action $S_{\text{N=2}}$: indeed $\delta_\text{CS} S_{\text{N=2}}=\int \eta(\mathbb A {\star} M\delta_\text{CS}\mathbb A+ \bbsigma {\star} M\delta_\text{CS}\bbsigma)$ which is linear in $\bblambda,\hat\bblambda$ so it may be cancelled using a new variation $\delta'$ satisfying $M\delta'\bblambda=M\delta'\hat\bblambda=0$. In fact, if $\delta'$ annihilates all bosons, it is determined uniquely this way. The complete transformations (on flat Euclidean space, with supersymmetry parameters $\zeta,\hat \zeta$ as before) are
\begin{subequations}
\label{eq:SUSYtransfsManin}
\begin{align}
    \delta\bbA_\mu &= -\zeta\gamma_\mu\bbhlambda + \bblambda\gamma_\mu\hzeta
    \Big.\,,\\
    \delta\bbsigma &= -i\zeta\bbhlambda + i\bblambda\hzeta
    \Big.\,,\\
    \delta\bblambda &= i\zeta\Big(\bbD + \bbsF - i\gsD\bbsigma {+ \frac{g^2}{k}M(i\bbsA - \bbsigma)}\Big)
    \Big.\,,\\
    \delta\bbhlambda &= \phantom{\zeta}i\Big(\bbD + \bbsF + i\gsD\bbsigma {+ \frac{g^2}{k}M(i\bbsA - \bbsigma)}\Big)\hzeta
    \Big.\,,\\
    \delta\bbD &= \gD_\mu\Big(\zeta\gamma^\mu\bbhlambda + \bblambda\gamma^\mu\hzeta\Big) + i\Big[\bbsigma,\zeta\bbhlambda + \bblambda\hzeta\Big]
    \Big.\,.
\end{align}
\end{subequations}
Since $M^2=0$ the new terms relative to the Chern-Simons SUSY are indeed annihilated by $M$.

Much like in $N=2$ Chern-Simons theory, the SUSY here is indeed evanescent: $\bbsigma,\mathbb D,\bblambda,\hat \bblambda$ may all be integrated out at the same time, leaving only the original Manin theory action.

\paragraph{Algebra.} The anticommutators of the above transformations are
\begin{subequations}
\label{eq:AnticommutatorsManin}
\begin{align}
    \{\delta_\zeta,\hdelta_\hzeta\}\bbA_\mu &= 2iK^\nu\Big(\bbF_{\nu \mu}  {+ \frac{g^2}{k}\epsilon_{\nu\mu\rho}M\bbA^\rho}\Big) + 2\zeta\hzeta\gD_\mu\bbsigma
    \Big.\,,\\
    \{\delta_\zeta,\hdelta_\hzeta\}\bbsigma &= 2iK^\nu\gD_\nu\bbsigma
    \Big.\,,\\
    \{\delta_\zeta,\hdelta_\hzeta\}\bblambda &= 2\big(iK^\nu\gD_\nu - \zeta\hzeta\ad\bbsigma\big)\bblambda  {+ \frac{g^2}{k}2M\bblambda(\zeta_\Cj\otimes\hzeta_\Cj)}
    \Big.\,,\\
    \{\delta_\zeta,\hdelta_\hzeta\}\bbhlambda &= 2\big(iK^\nu\gD_\nu - \zeta\hzeta\ad\bbsigma\big)\bbhlambda  {- \frac{g^2}{k}2(\zeta_\Cj\otimes\hzeta_\Cj)M\bbhlambda}
    \Big.\,,\\
    \{\delta_\zeta,\hdelta_\hzeta\}\bbD &= 2\big(iK^\nu\gD_\nu - \zeta\hzeta\ad\bbsigma\big)\Big(\bbD  {- \frac{g^2}{k}M\bbsigma}\Big)  {-\frac{g^2}{k}2\big(\zeta\hzeta\gD_\nu + iK_\nu\ad\bbsigma\big)M\bbA^\nu}
    \Big.\,,\\
        (\delta_\zeta)^2 &= (\delta_{\hat \zeta})^2 = 0\,.
\end{align}
\end{subequations}
where $\zeta,\eta$ and $\hat\zeta,\hat \eta$ are pairs of constant bosonic spinors of the same R-charge.

For $M=0$ we recognise the usual $N=2$ SUSY algebra given in \eqref{eq:SUSYALGEBRAMASSLESS}: the $M=0$ terms are organised into infinitesimal translations by $K^\mu\equiv \hzeta\gamma^\mu\zeta$ and gauge transformations with parameter $ K^\nu \mathbb A_\nu + i \zeta\hat\zeta\bbsigma$, which is valued in $\frak d$. The latter compensate for the gauge choice (Wess-Zumino gauge). Ignoring these gauge transformations we thus see that the supersymmetries square to translations, as expected.

What happens when $M\neq 0$? Recall that the gauge symmetry under $\frak d$ is explicitly broken to the subalgebra $\frak g\into\fd$. Therefore the gauge parameter $ K^\nu \mathbb A_\nu + i \zeta\hat\zeta\bbsigma$ includes not just $\frak g$-valued gauge transformations, which should be ignored as before, but also $\tfg$-valued transformations which are no longer gauge symmetries, but trivial symmetries (on-shell vanishing). The new $M$-dependent contributions to the algebra also ensure that these anticommutators vanish on-shell, which is required for consistency. The price to pay is that the transformations \eqref{eq:SUSYtransfsManin} \emph{no longer square to translations}, and the complete algebra of anticommutators has not been determined yet. 

\paragraph{Remarks.} 
\begin{enumerate}
    \item The supersymmetric action \eqref{eq:SUSYManinLagrangian} provides a supersymmetrisation for all 3D gauge theories that admit a Manin theory formulation. This includes Yang-Mills theory as shown in section \ref{sec:examples}. We emphasise that the action \eqref{eq:SUSYManinLagrangian} --- which enjoys evanescent $N=2$ supersymmetry --- is inequivalent to the usual $N=2$ super Yang-Mills action; this is another reason why it is convenient to have the adjective `evanescent' to distinguish the two situations. (In fact one could likely introduce an evanescent supersymmetry on top of conventional $N=2$ supersymmetry! We choose not to, though.)

    \item For the \trd{}  (see section \ref{sec:3rdway}) the action and supersymmetry transformations were originally found in the second author's Master's thesis \cite{Decavel:2023rmi}  (they were not derived from Manin theory).
\end{enumerate}

\subsubsection{SUSY on curved space; the round 3-sphere}
\label{sec:SUSYcurved}
For the purposes of localisation we will generalise the evanescent supersymmetry transformations \eqref{eq:SUSYtransfsManin} to curved space.  (The action remains \eqref{eq:SUSYManinLagrangian}.) For this, we will no longer assume that the supersymmetry parameters $\zeta,\hat \zeta$ are constant. 

We will in fact assume that $\zeta, \hat\zeta$ are \emph{conformal Killing spinors}\footnote{Their spinor bilinears are indeed conformal Killing vectors:
$\nabla_{(a}K_{b)} = \nabla_{(a}\zeta\gamma_{b)}\hat\zeta + \zeta\gamma_{(a}\nabla_{b)}\hat\zeta = \big(\zeta'\hat\zeta + \zeta\hat\zeta{'}\big)\eta_{ab}\,.$
}, i.e.
\begin{align}
&
\begin{gathered}
\exists\zeta': \nabla_\mu\zeta = \zeta'\gamma_\mu
\Big.\\
\exists\zetah{'}: \nabla_\mu\zetah = \gamma_\mu\zetah{'}
\Big.
\end{gathered}
&
&\Rightarrow
&&
\begin{gathered}
3\zeta' = \zeta\nablasl
\Big.\\
3\zetah{'} = \nablas\zetah
\Big.
\end{gathered}
&
&\Leftrightarrow
&&
\begin{aligned}
\zeta\big(\nablasl\gamma_\mu + 3\gamma_\mu\nablasl\big)\phantom{\zetah} &= 0
\Big.\\
\big(\gamma_\mu\nablas + 3\nablas\gamma_\mu\big)\zetah &= 0
\Big.
\end{aligned}
\end{align}
If $\zeta,\hat \zeta$ are conformal Killing, the transformations below are invariances of the action \eqref{eq:SUSYManinLagrangian} and satisfy $(\delta_\zeta)^2=(\delta_{\hat \zeta})^2=0$; the coefficients of the extra terms relative to \eqref{eq:SUSYtransfsManin} are fixed uniquely by these requirements:
\begin{subequations}
\label{eq:ManinSUSYCurved}
\begin{align}
\delta \mathbb A_\mu &= -\zeta\gamma_\mu\hat\bblambda+ \bblambda\gamma_\mu\zetah
\bigg.\\
\delta\bbsigma &= -i\zeta\hat \bblambda + i\bblambda\zetah
\bigg.\\
\delta\bblambda &= i\zeta\bigg(\mathbb D + \slashed{\mathbb F} - i\Dgs\bbsigma - \frac{2i}{3}\nablasl\bbsigma + \frac{g^2}{k}M(i\slashed{\mathbb A}
 - \bbsigma)\bigg)
\bigg.\\
\delta\hat\bblambda &= \phantom{\zeta}i\bigg(\mathbb D + \slashed{\mathbb F} + i\Dgs\bbsigma + \frac{2i}{3}\bbsigma\nablas + \frac{g^2}{k}M(i\slashed{\mathbb A} - \bbsigma)\bigg)\zetah
\bigg.\\
\delta \mathbb D &= \Dg_\mu\big(\zeta\gamma^\mu\hat\bblambda + \bblambda\gamma^\mu\zetah\big) + i\big[\bbsigma,~\zeta\hat\bblambda + \bblambda\zetah\big] - \frac{2}{3}\zeta\nablasl\hat\bblambda - \frac{2}{3}\bblambda\nablas\zetah
\bigg.
\end{align}
\end{subequations}
These transformations are similar to those of reference \cite{Closset:2012ru}, except in the presence of the $M$ terms in the variations of the fermions. (Also our conventions do not exactly match theirs.)

\paragraph{Supersymmetric Wilson loops.} A Wilson loop in the representation $\cal R$ is written 
\begin{equation}
\label{eq:susywilsonline}
W = \operatorname{Tr}_{\gR}\mathcal P\exp\oint_\gamma\dr\tau\Big[-\dot x^\mu \mathbb A_\mu + i|\dot x|\bbsigma\Big]
\end{equation}
where the factor of $i$ ensures that the supersymmetry preservation conditions
\be
\zeta\big(\dot x^\mu\gamma_\mu + |\dot x|\big) = 0\,,\qquad \& \qquad
\big(\dot x^\mu\gamma_\mu + |\dot x|\big)\hat \zeta = 0
\ee
are such that they admit non-zero solutions for at least one of $\zeta,\hat \zeta$. (We address reality issues later.) Note that the metric $g_{\mu\nu}$ makes its appearance via the curved gamma matrix and via $|\dot x|\equiv\sqrt{g_{\mu\nu}\dot x^\mu \dot x^\nu}$. These conditions are identical to those for undeformed $N=2$ supersymmetry, owing to how the transformations for $\mathbb A$ and $\bbsigma$ are not affected by the deformation involving $M$.

Fixing a specific spinor $\zeta$ with conjugate $\zeta^\dagger$ leads to the usual solution of the supersymmetry conditions, with the loop defined by the spinor bilinear $K$:
\be
\dot x=-K\,,\quad K^\mu\equiv \zeta\gamma^\mu \zeta^\dagger\,.
\ee
The corresponding Wilson loop is  invariant under $\delta_\zeta$.

\paragraph{The round $S^3$.} For the round 3-sphere of radius $\ell$, the covariant derivatives of spinors take the form (see appendix \ref{appendix:sphere})
\be
\label{eq:sphereSpinorD}
\nabla\zeta = \dr\zeta - \frac{i}{2\ell}\zeta\gamma \,,\quad
\nabla\hat \zeta = \dr\hat \zeta + \frac{i}{2\ell}\gamma\hat \zeta
\ee
Constant spinors are therefore conformal Killing; in fact, they are Killing spinors. and the spinor bilinear $K$ is a Killing vector. In that case the transformations \eqref{eq:ManinSUSYCurved} specialise to
\begin{subequations}
\label{eq:MANINSUSY3SPHERE}
\begin{align}
\delta \mathbb A_\mu &= -\zeta\gamma_\mu\hat\bblambda+ \bblambda\gamma_\mu\zetah
\bigg.\\
\delta\bbsigma &= -i\zeta\hat \bblambda + i\bblambda\zetah
\bigg.\\
\delta\bblambda &= i\zeta\bigg[(\mathbb D-\bbsigma/\ell) + \slashed{\mathbb F} - i\Dgs\bbsigma  + \frac{g^2}{k}M(i\slashed{\mathbb A}
 - \bbsigma)\bigg]
\bigg.\\
\delta\hat\bblambda &= \phantom{\zeta}i\bigg[(\mathbb D-\bbsigma/\ell) + \slashed{\mathbb F} + i\Dgs\bbsigma  + \frac{g^2}{k}M(i\slashed{\mathbb A} - \bbsigma)\bigg]\zetah
\bigg.\\
\delta \mathbb D &= \Dg_\mu\big(\zeta\gamma^\mu\hat\bblambda + \bblambda\gamma^\mu\zetah\big) + i\big[\bbsigma,~\zeta\hat\bblambda + \bblambda\zetah\big] +i(\zeta\hat\bblambda-\bblambda\hat\zeta)/\ell
\bigg.\,.
\end{align}
\end{subequations}

\subsubsection{Parity and Euclidean unitarity (reflection-positivity)}
\label{secParity}
\paragraph{Parity/reflection invariance.} Chern-Simons theories generically break parity invariance. However, we have already seen that Manin theory can give rise to familiar parity-invariant gauge theories, like Yang-Mills. There is a way to see parity-invariance \emph{a priori}, however. This is essentially the same as the parity-preservation mechanism of the \trd{} theory \cite{Arvanitakis:2015oga}: we augment the action of a space(time) orientation-reversing diffeomorphism with a transformation of colour indices.

Specifically: given a \emph{parity involution} $J :\fd\to \fd$ which by definition satisfies
\be
\label{eq:parityinvolutiondef}
J^2=1\,,\quad \eta(J\mathbbm{x},J\mathbbm{y})=-\eta(\mathbbm{x},\mathbbm{y})\,,\quad \text{and}\quad J[\mathbbm{x},\mathbbm{y}]=[J\mathbbm{x},J\mathbbm{y}]\,,
\ee
we may combine it with some orientation-reversing diffeomorphism $r:M_3\to M_3$ with $r^2=1$ to obtain a parity transformation $\mathcal P$
\be
\label{eq:ParityonA}
\bbA\to \mathcal P\bbA\equiv J r^\star\bbA
\ee
(with $r^\star$ the pullback by $r$) which preserves the sign of the Chern-Simons action:
\be
\begin{split}
\int_{M_3}\eta(\bbA \dr \bbA)\to -\int_{M_3}\eta(J\bbA\,\dr J\bbA)=\int_{M_3}\eta(\bbA \,\dr \bbA)\,,\\
\int_{M_3}\eta(\bbA^3)= \frac{1}{2}\int_{M_3}\eta(\bbA,[\bbA,\bbA])\to -\frac{1}{2}\int_{M_3}\eta(J\bbA,[J\bbA,J\bbA])=+\int_{M_3}\eta(\bbA^3)\,.
\end{split}
\ee
(We note here that $J$ also maps solutions of the Chern-Simons field equations to themselves, since $\bbF \to r^\star(\dr J\bbA+\frac{1}{2}[J\bbA,J\bbA])=r^\star J\bbF$.)

A parity involution exists for all Manin pairs coming from Lie quasi-bialgebras of coboundary type  \eqref{eq:coboundary}. This involution is given by
\be
\label{eq:parityinvolution}
JT_a=T_a\,,\qquad J\tT^a=-\tT^a +2 R^{ab}T_b\,.
\ee
It can be verified via direct calculation using the explicit expressions for the brackets in $\fd$ \eqref{eq:manindoublecommutationrelations} that this $J$ satisfies all conditions \eqref{eq:parityinvolutiondef}. Moreover it is trivial to check that 
\be
JM=-MJ=M
\ee
for this $J$ so that the bosonic Manin theory action \eqref{eq:ManinLag} is parity even under the transformation \eqref{eq:ParityonA} for any orientation reversing isometry $r$ of $M_3$.

The supersymmetric Manin theory defined by the action \eqref{eq:SUSYManinLagrangian} is \emph{also} parity-invariant. For this we specialise to Euclidean space $M_3=\mathbb R^3$ with the flat metric in Cartesian coordinates for the purposes of illustration. Under the reflection e.g.~$r(x^1,x^2,x^3)=(x^1,x^2,-x^3)$, the transformations of all fields are defined as
\be
\begin{split}
    \mathcal P\bbA =r^\star J\bbA\,,\quad \mathcal P \bbsigma=-r^\star J\bbsigma\,,\quad \mathcal P\mathbb{D}= r^\star J\mathbb D\,,\\
\mathcal P\bblambda=- J\bblambda \gamma^3\,,\qquad \mathcal P\hat{\bblambda}= J \gamma^3 \hat{\bblambda}\,.
\end{split}
\ee
These may be determined uniquely by demanding that the SUSY transformations \eqref{eq:SUSYtransfsManin} are such that 
\be
\mathcal P\delta_\zeta\bullet=\delta_{\mathcal P\zeta}\mathcal P\bullet
\ee
for all fields, where $\zeta$ and $\hat \zeta$ transform analogously to $\bblambda$ and $\hat \bblambda$. (The transformation of the spinors may be motivated as follows: given a vector $V$ in $\mathbb R^3$ and a unit vector $u$ defining a reflection along $u$'s perpendicular plane, the transformation of $V$ may be written via $\slashed{V}$ as $\slashed{V}\to-\slashed{u}\slashed{V}\slashed{u}$. Therefore expressions like $\bblambda \gamma_\mu \hat \zeta$ transform as vectors under e.g.~reflections about the $12$-plane when the spinors transform as stated above.)

Finally, we remark that it appears to be impossible to arrange that the Manin action is \emph{odd} under parity (instead of even) at least for simple gauge algebras, because of the difficulty of constructing a suitable involution on $\fd$.

\paragraph{Euclidean unitarity: reflection positivity.} We will consider unitarity at the  level of the path integral. The  brief discussion by Witten \cite[section 2]{Witten:1989ip} for reality conditions in complex Chern-Simons theory is relevant to us.

In both Lorentzian and Euclidean signatures unitarity entails  probabilities, partition functions, \dots, calculated in the quantum theory are non-negative. For Lorentzian signature the requirement is that the integrand $e^{iS}$ of the path integral lies in $\mathrm{U}(1)$, which implies the Lorentzian action $S$ must be real.

We are more interested in Euclidean unitarity. The requirement is the following: whenever we reverse spacetime orientation, the integrand should be complex conjugated. When this is the case, we may formally see that e.g.~the partition function on $\mathbb R^3$ is non-negative via a cutting and gluing argument\footnote{For more on cutting and gluing for non-topological QFTs we refer to work by Dedushenko \cite{Dedushenko:2018aox}.}. To see this, cut $\mathbb R^3$ along the 2-plane $D$ defined by $x^3=0$, so $\mathbb R^3=\mathbb R^3_+\cup_D\mathbb R^3_-$. Then if we define the wavefunctionals $\Psi_\pm$ on each half-space $\mathbb R^3_\pm$ defined by the boundary conditions $\phi|_D=\pi$ of all fields $\phi$\,,
\be
\Psi_\pm(\pi)\equiv \int_{\phi|_D=\pi}\mathcal D\phi \;\exp(-S_{\mathbb R^3_\pm})\,,
\ee
($S_{\mathbb R^3_\pm}$ being the action integrated over the respective half-space) we may express the partition function as their overlap by integrating over all possible boundary values:
\be
Z=\int\mathcal D\phi \;e^{-S_{\mathbb R^3}}=\int\mathcal D\phi\; e^{-S_{\mathbb R^3_+}}e^{-S_{\mathbb R^3_-}}=\int \mathcal D\pi\;\Psi_+(\pi)\Psi_-(\pi)\,.
\ee
Since the reflection that sends $x^3\to -x^3$ maps $\mathbb R_-^3\to\mathbb R_+^3$, if $S_{\mathbb R_-^3}$ goes to $S_{\mathbb R_+^3}^\ast$ --- the complex conjugate of $S_{\mathbb R_+^3}$ --- we see the right-hand side is, insofar as the path integral over $\pi$ is defined, non-negative. (This argument may be viewed as inserting 1 in the matrix element $Z=\langle 0|0\rangle$.)

\paragraph{The reality condition on fields and a subtlety involving Wilson loops.}

As established previously, the supersymmetric Manin action \eqref{eq:SUSYManinLagrangian} is reflection-invariant. Therefore the simplest reality condition is one where the action is real. Let us discuss the real case. The obvious choice is to take $\mathbb A,\bbsigma,\mathbb D$ to be real $\fd$-valued fields, and to assume the constants $k$ and $g^2$ are real. However, there arises a reality issue with Wilson loops, which to our knowledge has not been discussed before. Given the transformations \eqref{eq:SUSYtransfsManin}, SUSY Wilson loops  take the form \eqref{eq:susywilsonline} whose integrand is $-\dot x^\mu \mathbb A_\mu + i|\dot x|\bbsigma$
with $|\dot x|=\sqrt{g_{\mu\nu} \dot x^\mu \dot x^\nu}$.  If $\bbsigma$ is real-valued, the integrand takes values in the complexification of the Lie algebra $\fd$. This may be cured if we assume that $\bbsigma$ is pure imaginary, instead; reflection positivity then demands that $\mathbb D$ is also pure imaginary. We then have two possible reality conditions: one where $\bbsigma,\mathbb D$ are both $\fd\otimes i\mathbb R$-valued, and one where they are both $\fd$-valued. ($\mathbb A$ must be $\fd$-valued in both cases.) Unfortunately the boundary condition where $\bbsigma$ is imaginary-valued is incompatible with the localisation argument given later  (it breaks the calculation leading to \eqref{eq:localisingfunctionalExplicit}).  \textbf{We thus employ the reality condition where $\mathbb A,\bbsigma$ and $\mathbb D$ are all  $\fd$-valued.} 

We were not able to find any other unitary branches, since we have not found a prescription for reflection such that the action is odd. This is in contrast to $G_{\mathbb C}$ Chern-Simons theory, where there is a second unitary branch, which in our conventions would entail pure imaginary $k$ \cite{Witten:1989ip}; this branch is not compatible with the Manin mass term.

Finally, we point out that while the parity involution $J$ of \eqref{eq:parityinvolution} exists for coboundary Lie quasi-bialgebras, it may also exist under more general circumstances. For example, when $\tfg$ may (and is chosen to) be  such that $(\tfg\into\fd,\eta)$ is another Manin pair, the Lie quasibialgebra structure on $\tfg$ need not be coboundary, but the involution $J$ of \eqref{eq:parityinvolution} that is defined relative to its `dual' Lie quasibialgebra $\fg$ still works. An explicit example is given by the Lie quasibialgebra defining Freedman-Townsend theory (sec.~\ref{sec:freedmantownsend}).

\section{Localisation of Manin gauge theory}
\subsection{BRST gauge fixing}

Manin theory is gauge invariant under $G$-valued gauge transformations --- the Lie group integrating $\frak g$. Therefore  the whole gang of ghosts, antighosts, and friends will be valued in $\frak g$ (as opposed to the bigger algebra $\frak d$):
\be
c=c^a T_a\,, \qquad \bar c=\bar c^a T_a\,,\qquad b= b^a T_a\,.
\ee
The BRST transformations $\delta_{\text{BRST}}$ take the standard form, keeping into account that this is a right differential like the supersymmetry is: 
\begin{subequations}
\begin{align}
    \delta_{\text{BRST}}\bbA_\mu &= -\gD_\mu c\,,
    &
    \delta_{\text{BRST}}c &= -c^2\,,
    \\
    \delta_{\text{BRST}}\bbsigma &= [c,\bbsigma]\,,
    &
    \delta_{\text{BRST}}\bc &= ib\,,
    \\
    \delta_{\text{BRST}}\bblambda &= -\{c,\bblambda\}\,,
    &
    \delta_{\text{BRST}}b &= 0\,,
    \\
    \delta_{\text{BRST}}\hat\bblambda &= -\{c,\hat\bblambda\}\,,
    \\
    \delta_{\text{BRST}}\bbD &= [c,\bbD]\,.
\end{align}
\end{subequations}
The BRST differential $\delta_{\text{BRST}}$ anticommutes with the supersymmetries\footnote{This requires a short calculation analogous to the one for conventional $N=2$ supersymmetry. A key point in the present case is the identity $M[x,\mathbbm y]=[x,M\mathbbm y]$ valid for all $\mathbbm y\in \fd$ and all $x\in \fg$ whenever $M$ satisfies \eqref{eq:Midentities}.} $\delta_\zeta$ and $\hat \delta_{\zeta}$ (assuming $\delta_{\BRST}$ annihilates the ghost sector); it is also nilpotent: $\delta_{\text{BRST}}^2=0$.

For the gauge-fixing action we need to pick a $\tfg$ complementary to $\frak g$ in $\frak d$. In the basis adapted to this split (used already in \eqref{eq:manindoublecommutationrelations}) the gauge-fixing action for Lorentz gauge reads (where $\nabla$ is the Levi-Civita derivative and $ D_A\equiv \dr +[A,\bullet]=\mathcal D-[\tilde A,\bullet]$)
\begin{equation}
\label{BRST action}
\begin{split}
    S_\BRST[\bbA,b,c,\bc]
    &=
    -\delta_{\text{BRST}}\int\star\kappa_{ab}\nabla^\mu A_\mu^a\bc^b
    = -\int\star\kappa_{ab}\bigg(ib^a\nabla^\mu A_\mu^b + \pd_\mu\bc^a\gD^\mu c^b\bigg)
    \\
    &= -\int\star\kappa_{ab}\bigg(ib^a\nabla^\mu A_\mu^b + \pd_\mu\bc^aD^\mu_{A} c^b {- \tf^{cb}{_d}\tA_c^\mu\pd_\mu\bc^ac^d}\bigg)\,.
\end{split}
\end{equation}
The difference with conventional (formulations of) gauge theory lies only in the last term that depends on $\tilde f$. We may select $\tfg$ such that $\tilde f=0$ whenever $\fg$ is of coboundary type as a Lie quasibialgebra, see \eqref{eq:coboundary} and discussion in that section. In particular we may choose this for the Manin formulations of Yang-Mills theory and Third Way theory. Since coboundary-type quasibialgebras are the ones for which we have established reflection-positivity (see \eqref{eq:parityinvolution}), \textbf{we will henceforth specialise to $\tilde f^{ab}{}_c=0$}.

In this situation we see explicitly, given the split $\fd=\fg+\tfg$, that the field $\mathbb A$ splits into a gauge field $A$ for $\frak g$ and a matter field $\tilde A$:
\be
\delta_\text{BRST} A_\mu = -\pd_\mu c -[A_\mu,c]\,,\qquad \delta_{\text{BRST}}\tilde A=-[\tilde A,c]\,.
\ee

\subsection{Localisation on $S^3$}
We will localise the gauge-fixed theory with respect to the (right-)differential
\be
\label{eq:Qdef}
Q\equiv \delta_{\text{BRST}}+\delta_\zeta
\ee
for $\delta_\zeta$ the supersymmetry transformation given in \eqref{eq:MANINSUSY3SPHERE} for the round 3-sphere with $\hat \zeta=0$ and $\zeta$ some fixed constant nonzero spinor. The action of the gauge-fixed theory ought to be $Q$-invariant, which motivates changing the gauge fixing action from its conventional form $S_\text{BRST}$ above to
\be
Q V_\text{g.f.}\equiv Q \Big(-\int\star\kappa_{ab}\nabla^\mu A_\mu^a\bc^b\Big)=S_\text{BRST}+\delta_\zeta \Big(-\int\star\kappa_{ab}\nabla^\mu A_\mu^a\bc^b\Big)\,,
\ee
so that the e.g.~partition function is the integral
\be
\label{partitionfunction}
Z= \int\mathrm D\bbA\, \mathrm D\bblambda\, \mathrm D \hat\bblambda\, \mathrm D\bbsigma\, \mathrm D \mathbb D\, \mathrm Dc\,\mathrm D \bar c\,\mathrm D b\, \exp(S_{N=2}+QV_\text{g.f.})
\ee
for $S_{N=2}$ the evanescent supersymmetric Manin theory action \eqref{eq:SUSYManinLagrangian}. Localisation will, of course, work for the correlation functions of any collection of $Q$-closed observables; in particular, for any observables which are SUSY- and BRST-invariant.

We will assume there are no supersymmetry or BRST anomalies, so that we may integrate by parts in the path integral with respect to $Q$. We may then freely subtract from  the action the following $Q$-exact term 
\be
Q\Psi\equiv Q\int \star \mathcal E_{AB}(\bblambda^A(\delta_\zeta \bblambda^B)^\dagger)
\ee
where $\mathcal E:\frak d\times \frak d\to \mathbb R$ is some positive semidefinite and $\ad \frak g$ invariant symmetric form; such $\mathcal E$ always exist by an averaging argument as long as $G$ is compact. If $\tilde f=0$ we may verify directly, for example, that $\mathcal E(T_a,T_b)=\kappa_{ab}\,,\mathcal E(\tilde T^a,\tilde T^b)=\kappa^{ab}$ works.

Since $Q$ is a right differential, we have \begin{equation}
\label{Localising functional}
Q\Psi = \int\star\gE\Big[(\delta_\zeta\bblambda)(\delta_\zeta\bblambda)^\dagger + \bblambda\delta_\zeta(\delta_\zeta\bblambda)^\dagger\Big]
\end{equation}
which evaluates to
\begin{equation}
\label{eq:localisingfunctionalExplicit}
\begin{aligned}[b]
Q\Psi
=
(\zeta\zeta^\dagger)\int\star\gE\bigg[\Big\|\star\bbF - \gD\bbsigma + \frac{g^2}{k}M\bbA\Big\|^2 + \Big(\mathbb D-\big(\frac{1}{\ell}+\frac{g^2}{k} M\big)\bbsigma\Big)^2
\\
+ 2\bblambda\Big(-i\gsD + \ad\bbsigma + \frac{g^2}{k}M + \frac{1}{2\ell}\Big)\bbhlambda
\end{aligned}
\bigg]
\end{equation}
This is positive-definite: the first term is the norm squared  of the vector field ${\star}\mathbb F-\mathcal D\bbsigma + M\mathbb A$ with respect to the 3-sphere metric,  while the second term is manifestly a square. (Here we have employed the reality condition on fields that establishes reflection positivity, namely that the fields all be real-valued.) 

We then modify the action to
\be
\label{eq:totalActionLoc}
S_{N=2}+QV_\text{g.f.}-t^2Q\Psi
\ee
inside the path integral, in standard localisation argument fashion; in the limit $t\to \infty$ the only configurations which contribute to the path integral must solve the equations
\begin{subequations}
\label{localisationequations}
\begin{align}
\label{Alocequation}
{\star}\mathbb F-\mathcal D\bbsigma +\frac{g^2}{k} M\mathbb A=0\,,\\
\label{Dlocequation}
\mathbb D-\Big(\frac{1}{\ell}+\frac{g^2}{k} M\Big)\bbsigma=0\,.
\end{align}
\end{subequations}
We have thus shown that \textbf{all Manin theories localise onto \eqref{localisationequations}} subject only to the assumptions that
\begin{itemize}
    \item the group $G$ integrating the gauge algebra $\fg$ is compact; and
    \item the differential $Q$ is an invariance of the path integral measure (i.e.~there are no supersymmetry or BRST anomalies).
\end{itemize}

\paragraph{The 1-loop determinant.} It is straightforward to write down an expression for the fluctuation determinants around the localisation locus. We expand the fields $\{\bbA,\bbsigma,\bbD,\bblambda,\hat\bblambda\}$ onto `moduli' $\{\bar \bbA,\bar\bbsigma,\dots\}$ and fluctuations $\{\bbA',\bbsigma',\dots\}$ as follows:
\begin{align}
\bbA&=\bar \bbA + t^{-1}\bbA'\,,\\
\bbsigma&=\bar \bbsigma  +  t^{-1} \bbsigma'\,,\\
\vdots&\nonumber
\end{align}
where $t$ is the same constant appearing in the total action \eqref{eq:totalActionLoc} (which includes gauge-fixing and localisation terms). The bosonic moduli are constrained to satisfy the localisation conditions \eqref{localisationequations} and the fermionic ones $\bblambda,\hat\bblambda$ are set to zero: therefore
\be
\bar \bbD=\left(\frac{1}{\ell} +\frac{g^2}{k}M\right)\bar\bbsigma\,,\qquad \bar\bblambda=\bar{\hat\bblambda}=0\,.
\ee
We do nothing to the gauge sector fields $\{c,\bar c,b\}$; in particular $\bar c$ remains the antighost. We should also think of $c,\bar c$ as `fluctuation' fields.

We perform the path integral $\int \mathrm Db$ \emph{before} taking the limit $t\to \infty$, as we should. Then the term $i b\nabla \cdot A$ from $S_\text{BRST}$ inside the gauge-fixing term $V_\text{g.f.}$ in the total action \eqref{eq:totalActionLoc} enforces the gauge $\nabla \cdot \bar A + t^{-1} \nabla \cdot A'=0$ for any value of $t$; thus both the modulus $\bar A$ and fluctuation $A'$ are gauge-fixed.\footnote{This procedure is thus slightly different from that of e.g.~Kapustin et al.~\cite{kapustin2010exact}. Our approach has the advantage that the localising functional is positive-definite, whereas including a gauge-fixing term in $\Psi$ would produce a contribution $i b \nabla\cdot \bar A$.}  The total action \eqref{eq:totalActionLoc} has a smooth limit $t\to \infty$ which is the sum
\be
\bar S_{N=2} + S_\text{1-loop}
\ee
where $\bar S_{N=2}$ is \eqref{eq:SUSYManinLagrangian} evaluated on the moduli and $S_\text{1-loop}$ is calculated easily via $t$-power-counting to be
\be
\begin{aligned}[b]
S_\text{1-loop}=-\int \star\kappa (\nabla_\mu \bar c D_{\bar A}^\mu c)\\
-(\zeta\zeta^\dagger)\int\star\gE\bigg(\Big\|\star\bar{\mathcal D} \bbA' - \bar{\mathcal D}\bbsigma' +[\bar\bbsigma,\mathbb A'] + \frac{g^2}{k}M\bbA'\Big\|^2 + \Big(\mathbb D'-\big(\frac{1}{\ell}+\frac{g^2}{k} M\big)\bbsigma'\Big)^2
\\
+ 2\bblambda'\Big(-i\bar\gsD + \ad\bar\bbsigma + \frac{g^2}{k}M + \frac{1}{2\ell}\Big)\bbhlambda'\bigg)\,.
\end{aligned}
\ee 
Barred quantities are evaluated with the moduli fields, so e.g.~$D_{\bar A}\equiv \dr + [\bar A,\bullet], \bar {\mathcal D}\equiv \dr + [\bar \bbA,\bullet]$, and $(\zeta\zeta^\dagger)$ is an arbitrary positive normalisation constant that may be set to 1 by rescaling the spinor. 

The localised expression for the partition function $Z$ is, therefore,
\be
\label{eq:localisedpartitionfunction}
Z= \int_{\mathcal M}\mathrm D\bar\bbA\,  \mathrm D\bar\bbsigma\,\delta(\nabla\cdot \bar A) Z_\text{1-loop} \exp(\bar S_{N=2})
\ee
for $\mathcal M$ the locus of solutions to the localisation equation \eqref{Alocequation}, where $Z_\text{1-loop}$ is the 1-loop determinant
\be
\label{eq:1loopdeterminant}
Z_\text{1-loop}=\int \mathrm D \bbA'\,\mathrm D \bbsigma'\, \mathrm D \bbD'\, \mathrm D \bblambda'\, \mathrm D\hat\bblambda'\, \mathrm D c\, \mathrm D\bar c \, \delta(\nabla\cdot  A')\exp{S_\text{1-loop}}\,.
\ee
The same argument works to localise any $Q$-invariant observable, e.g. a SUSY Wilson loop.

We will not attempt to evaluate $Z_\text{1-loop}$ in this work, but we have more to say about the integral over the moduli space $\mathcal M$.

\subsection{Resolving the localisation locus $\mathcal M$}
Henceforth we remove the bars over the moduli fields $\bar \bbA,\bar \bbsigma,\dots$ since we no longer discuss the fluctuation fields $\bbA',\bbsigma'$ in the 1-loop determinant.

The integration over the field $\mathbb D$ is eliminated due to \eqref{Dlocequation} in all cases. However at this stage the localisation is more complicated than for Chern-Simons theory: the presence of the $M\bbA$ term in \eqref{Alocequation} would appear to suggest, a priori, that the integration over $\bbsigma$ is infinite-dimensional. This would be problematic insofar as it means that we have replaced the infinite-dimensional integral --- over $\bbA$, for the $N=0$ theory --- with an equivalent, yet even more infinite-dimensional integral.

That may well be the case for arbitrary Manin pairs. Nevertheless, we may try to integrate out $\bbsigma$ in favour of a finite-dimensional zeromode integral in specific cases of physical interest, including the Third Way theory and Yang-Mills theory. We treat each case separately.

\subsubsection{Third Way theory}
This is the case of the Manin pair whose double $\fd$ is the (commuting) direct sum of Lie algebras $\fd=\fg\oplus \fg$, with the maximally isotropic subalgebra $\fg$ embedded as the diagonal as summarised in section \ref{sec:3rdway}. The key feature of this Manin pair is that the double $\fd$ integrates to a compact Lie group $\mathbb D$, which is (a cover of) $G\times G$. (Since the scalar field $\mathbb D$ has already been integrated out there is no notational clash.) Therefore the form $\mathcal E$ employed above may be chosen to be invariant not just under the $G$ subgroup of $\mathbb D$ --- which would be the diagonal $G$ for $\mathbb D=G\times G$ --- but also under the full group $\mathbb D$; we commit to such a choice of $\mathcal E$ in this subsection. 

Hitting the localisation equation \eqref{Alocequation} with $\star\mathcal D\star$ leads to
\be
\label{eq:thirdwaySigmaPDE}
{\star}\mathcal D{\star}\mathcal D\bbsigma=\frac{g^2}{k}{\star}\mathcal D M{\star}\mathbb A\,,
\ee
where again $\mathcal D= \dr + \ad\mathbb A$. We will solve this for $\bbsigma$ in terms of $\mathbb A$ up to zeromodes, and establish that the space of zeromodes is finite-dimensional. (This is nontrivial because $\mathbb A$ is not a flat connection.) The key to this is establishing that $$\Delta_{\mathbb A}\equiv{\star}\mathcal D{\star}\mathcal D:C^\infty( M_3)\otimes \frak d\to C^\infty( M_3)\otimes \frak d$$ is both elliptic and selfadjoint; then one may use standard theorems (e.g.~Theorem 4.12 of \cite{wells2007differential}) to get a handle on the space of zeromodes and the existence of Green's functions.

Ellipticity is in fact trivial: the principal symbol of $\Delta_{\mathbb A}$ is independent of $\mathbb A$, and for $\mathbb A=0$ this operator reduces to the Laplacian. For self-adjointness we invoke $\mathcal E:\frak d\to \frak d$ as above and define an $L^2$ inner product
\be
\langle\bbsigma,\bbtau\rangle\equiv\int_{S_3} \mathcal E(\bbsigma,{\star}\bbtau)
\ee
such that  an operator $\mathcal O$ is selfadjoint if and only if
\be
\langle \mathcal O\bbsigma,\bbtau\rangle=\langle\bbsigma,\mathcal O \bbtau\rangle\,.
\ee
If $\mathcal D$ satisfies a standard integration by parts identity inside the integral, then $\Delta_{\mathbb A}$ is selfadjoint by a small calculation. That identity is indeed satisfied whenever $\mathcal E$ is $\ad\frak d$-invariant, which is true by construction of $\mathcal E$.

Since $\Delta_{\mathbb A}$ is elliptic and selfadjoint, Theorem 4.12 of \cite{wells2007differential} implies a Hodge decomposition: 
\begin{enumerate}
    \item there exists a projector $\Pi_H$ onto $H=\ker \Delta_{\mathbb A}$ (``harmonic scalars'') and an operator $G$ (``Green's function'') such that
    \be
    G \Delta_{\mathbb A}+\Pi_H=\Delta_{\mathbb A} G+\Pi_H=1
    \ee
    with $1$ being the identity on $C^\infty( M_3)\otimes \frak d$.
    \item $H$ and $\im(G\Delta_{\mathbb A})$ are orthogonal with respect to the $L^2$ inner product;
    \item
    the dimension of $H$ is finite.
\end{enumerate}
For us $H$ is the space of $\bbsigma$ zeromodes and $N_{\mathbb A}\equiv \im(G\Delta_{\mathbb A})$ is the space of non-zero modes. To solve \eqref{eq:thirdwaySigmaPDE} using the Green's function we thus need to confirm that the source term $\frac{g^2}{k}{\star}\mathcal D M{\star}\mathbb A$ actually lies in $N_{\bbA}$. Indeed it is trivial to check that it is orthogonal to any zeromode $\bbsigma_0$:
\be
\int_{S^3}\mathcal E({\star}\mathcal DM{\star}\mathbb A,{\star}\bbsigma_0)=\int_{S^3}\mathcal E(\bbsigma_0,\mathcal DM{\star}\mathbb A)=0
\ee
where the last equality uses integration by parts and the implication $\Delta_{\mathbb A}\bbsigma_0=0\iff \mathcal D\bbsigma_0=0$.

Therefore we have established that solutions $\bbsigma$ to \eqref{eq:thirdwaySigmaPDE} take the form $\bbsigma=\bbsigma_0+\bbsigma'$ where $\bbsigma_0$ lies in a finite-dimensional space of zeromodes and the non-zero mode $\bbsigma'$ is uniquely determined in terms of $\mathbb A$, $M$, and the Green's function $G$ (which also depends on $\mathbb A$) as
\be
\label{eq:bbsigmaprimeThirdWaySol}
\bbsigma'=\frac{g^2}{k} G {\star}\mathcal D {\star} M\bbA\,.
\ee
($\bbsigma'$ here not to be confused with the fluctuation field from the previous subsection.) 

\paragraph{The space of zeromodes and the localised path integral.} The zeromodes $\bbsigma_0$ at fixed $\mathbb A$,
\be
\mathcal D\bbsigma_0=\dr\bbsigma_0+[\mathbb A,\bbsigma_0]=0\,,
\ee
form the space of $\fd$-valued infinitesimal gauge transformations that leave $\mathbb A$ invariant. This is the tangent space to the \emph{stabiliser} of $\mathbb A$, namely the group of finite gauge transformations leaving $\mathbb A$ invariant,
\be
S_{\mathbb A}\equiv \big\{\mathbbm g:S^3\to \mathbb D \quad | \quad\mathbb A= -\dr \mathbbm g\mathbbm g^{-1} + \mathbbm g \mathbb A \mathbbm g^{-1}\big\}\,,
\ee
where $\mathbb D=G\times G$ indicates the group integrating $\fd$.

Although $\mathbb D$ gauge transformations are not a symmetry of Manin  theory ---  indeed we have been emphasising that $\mathbb A$ is the data of a gauge field for $G$ along with a matter field ---  we may use gauge theory results to study the zeromodes for this specific Manin theory\footnote{Note that when $\mathbb D$ is simply-connected, since $S^3$ is 3-dimensional all principal bundles atop $S^3$ with fibre $\mathbb D$ are trivial, so there are no subtleties in identifying $\mathbb A$ with a $\mathbb D$-connection for the purposes of the argument in this paragraph. $\mathbb D$ may be arranged to be simply-connected if e.g.~we take $\mathbb D=G\times G$ for simply connected $G$.}, which are conveniently collected for us in reference \cite{Fuchs:1994zv}. It is known in particular that $S_A$ may be identified with the centraliser of the holonomy group of $\mathbb A$, which is a subgroup of $\mathbb D$. Therefore the dimension of the space of zeromodes $N_{\mathbb A}$, which is the dimension $\dim S_{\mathbb A}$, is bounded above by $2\dim G$ and depends on the specific $\mathbb A$. 

We return to the localisation equation \eqref{Alocequation}, ${\star}\mathbb F-\mathcal D\bbsigma +\frac{g^2}{k} M\mathbb A=0$. Since it is only the nonzero modes of $\bbsigma$ that appear therein, we may eliminate $\bbsigma$ in favour of $\mathbb A$ using the Green's function $G$ and rewrite \eqref{Alocequation} as
\be
\label{AlocequationThirdWay}
{\star}\mathbb F +\frac{g^2}{k} M\mathbb A=\mathbb J_{\mathbb A}
\ee
for a source term $\mathbb J_{\mathbb A}\equiv \mathcal D\bbsigma=(g^2/k)\mathcal DG {\star}\mathcal D{\star} M\mathbb A$ that depends on $\mathbb A$ nonlinearly (via $M\mathbb A$, $\mathcal D$ and $G$).

Upon using the above results, the localised path integral \eqref{eq:localisedpartitionfunction} for the partition function $Z$ becomes
\be
\label{eq:localisedThirdWayPartitionFunction}
Z= \int_{\mathcal M_\text{Third Way}}\mathrm D\bbA\,\delta(\nabla\cdot  A)\int_{N_{\bbA}}  \mathrm d\bbsigma_0\, Z_\text{1-loop} \exp(\bar S_{N=2})
\ee
This is an integral over the \emph{a priori infinite-dimensional} moduli space $\mathcal M_\text{Third Way}$  of solutions $\bbA$ to  equation \eqref{AlocequationThirdWay}, 
along with an integral over the finite-dimensional space of zeromodes $\bbsigma_0\in N_{\bbA}$ associated to each such solution. Note that $\mathcal M_\text{Third Way}$ is non-empty: it includes e.g.~$\mathbb A=0$ which trivially solves \eqref{AlocequationThirdWay}.

Therefore the path integral of the Third Way theory localises to the classical equations of motion with a source $\mathbb J_{\mathbb A}$ generated by self-interactions, alongside a finite integral over zeromodes $\bbsigma_0$.

\subsubsection{Yang-Mills theory}
As explained in section \ref{Yang-Mills as a Manin theory}, this is the case where the double $\fd$ is the semidirect sum of $\fg$ and its coadjoint representation: $\fd=\fg\ltimes \fg^\star$. We expand everything in the basis $\{T_a,\tilde T^a\}$ of \eqref{eq:YMdoublealgebra} exhibiting this semidirect sum structure  immediately:
\be
\mathcal D\bbsigma=D_A\sigma  + D_A\tilde \sigma +[\tilde A,\sigma]\,,\quad \mathbb F= F+D_A\tilde A\,,
\ee
with $F\equiv\dr A + A^2$, and $D_A\equiv \dr +[A,\bullet]$ the covariant derivative, where the bracket is in $\fd$. Since $[T,\tilde T]\propto \tilde T$ and the bracket restricted to $\fg$ closes, all quantities with tildes lie in $\tilde{\fg}\equiv \fg^\star$. The localisation equation \eqref{Alocequation} splits into two components, each valued in $\fg$ and $\tilde{\fg}$ respectively:
\begin{align}
\label{AlocYM7}
\frac{g^2}{k}M\tilde A&=D_A\sigma-{\star}F\,,\\
\label{AlocYM8}
{\star}D_A\tilde A&=D_A\tilde \sigma +[\tilde A,\sigma]\,.
\end{align}
These two equations are respectively the $\tilde A$ equation of motion of the first-order Yang-Mills action \eqref{first order YM} sourced by $D_A\sigma$, and the $A$ equation of motion of the same sourced by $D_A\tilde \sigma +[\tilde A,\sigma]$.

Since $M\tilde A=M^{ab}\tilde A_b T_a$ with $M^{ab}$  nondegenerate (see discussion around \eqref{eq:Midentities}), \eqref{AlocYM7} may be solved for $\tilde A$. If we also simplify the notation using $M^{ab}$ to identify the coadjoint representation $\fg^\star=\tilde {\fg}$ with $\fg$ everywhere,  \eqref{AlocYM8} becomes
\be
\label{AlocYMFinal}
{\star} D_A {\star }F= 2{\star}[F,\sigma]-\Big( \frac{g^2}{k} D_A\tilde\sigma +[D_A\sigma,\sigma]\Big)
\ee
which is the second-order Yang-Mills equation of motion with sources. Note that these are not the equations of motion one obtains from $S_{N=2}$ \eqref{eq:SUSYManinLagrangian} with $\tilde A$ backsubstituted in, because in that case $S_{N=2}$ takes the form
\be
\frac{k^2}{g^2}\int\kappa\Big(\frac{1}{2} D_A\sigma\star D_A\sigma -\frac{1}{2} F\star F\Big)-\Big(\frac{g^2}{2} +\frac{k}{\ell}\Big)\int \kappa\Big(\frac{1}{2}\tilde \sigma\star\tilde\sigma\Big) +\text{(fermions)}
\ee

We can solve equation \eqref{AlocYMFinal} for the nonzero modes of $\tilde \sigma$ using a propagator, as was done in the previous subsection. Thus we split $\tilde \sigma=\tilde\sigma_0+\tilde\sigma'(A,\sigma)$ with $\tilde\sigma_0$ satisfying $D_A\tilde\sigma_0=0$. The upshot is the following formula for the Yang-Mills partition function \eqref{partitionfunction}:
\be
\label{eq:localisedYMManinPartitionFunction}
Z= \int_{\mathcal M}\mathrm DA\,\mathrm D\sigma\,\delta(\nabla\cdot A)\,\int_{N_{A,\sigma}}\mathrm d\tilde\sigma_0\, Z_\text{1-loop} \exp{\bar S_{N=2}}\,,
\ee
where now $\mathcal M_\text{Yang-Mills}$ is the moduli space of solutions to \eqref{AlocYMFinal}, while $Z_\text{1-loop}$ is the usual localisation-induced 1-loop determinant. 

In contrast to the Third Way case, and, perhaps, predictably, we have not replaced the usual $N=0$ path integral with something that is obviously simpler: we have  integrals over $A$ and $\sigma$, both of which are a priori infinite-dimensional even though these fields are related via \eqref{AlocYMFinal}. We note however that this result does imply that the path integral is 1-loop exact for abelian $G$: using a Hodge decomposition, formula \eqref{AlocYMFinal} implies immediately that both $F$ and $\tilde \sigma$ are harmonic, so the integral has localised onto $\dr F=\dr {\star} F=0$; the other fields contribute normalisation factors.

\section{Discussion}
For the convenience of the reader we recapitulate the main results of this paper as well as the specific technical conditions invoked:
\begin{enumerate}
    \item For any Manin gauge theory where the gauge algebra $\fg$ is associated to a compact gauge group $G$, and assuming that there are no BRST or supersymmetry anomalies, expectation values of $Q$-invariant  observables (see \eqref{eq:Qdef}) on the round 3-sphere $S^3$ localise onto \eqref{localisationequations}. The resulting path integral takes the form \eqref{eq:localisedpartitionfunction} (for the partition function) where the integral measure acquires the 1-loop determinant factor \eqref{eq:1loopdeterminant}. (These results were obtained when the Lie quasibialgebra associated to the Manin pair has $\tilde f=0$, which may be arranged to be the case for the large class of Lie quasibialgebras of coboundary type \eqref{eq:coboundary}; however, this restriction can be lifted.)
    \item We proved the Hamiltonian \eqref{eq:Maninhamiltonian} is positive-definite in Lorentzian signature (in nice enough coordinate systems).
    \item We also showed that Manin gauge theories whose Lie quasibialgebras are of coboundary type --- including the case $\tilde f=0$ --- are all parity-invariant and reflection-positive.
    \item For the Manin theory formulations of Yang-Mills and Third Way theories (which fulfil all assumptions used above) we further analysed the localisation locus, yielding the path integrals \eqref{eq:localisedYMManinPartitionFunction} and \eqref{eq:localisedThirdWayPartitionFunction} respectively.
\end{enumerate}

One should in principle check carefully for anomalies in the realisation of evanescent supersymmetry in Manin theory, since it differs by terms depending on the mass matrix $M$ from conventional 3D $N=2$ supersymmetry. In fact such anomalies could also afflict conventional localisation calculations, as has been pointed out multiple times in the literature (see e.g.~\cite{Katsianis:2019hhg} and references therein). Although it is reassuring that R-symmetry anomalies are forbidden on dimensionality grounds, the current absence of a superspace formulation, as well as the fact that the complete algebra of evanescent supersymmetries is unknown, signals, perhaps, a need to revisit and generalise the literature on supersymmetry anomalies.

It is worth discussing our results on localised partition functions for Yang-Mills and Third Way theories. Even though the Third Way theory is a  deformation of Yang-Mills theory, the localised path integrals look qualitatively different. In the Yang-Mills case, the locus \eqref{AlocYMFinal} looks like a nonlinear duality relation between the gauge connection $A$ and the scalar field $\sigma$ that appear in the localised path integral, while in the Third Way case the localised path integral is over the original fields $A,\tilde A$, constained by the Third Way equations of motion with a nonlocal self-interaction term \eqref{AlocequationThirdWay} (along with a finite-dimensional path integral, which we ignore). 

In other words, the Third Way theory is unreasonably close to being 1-loop exact! Clearly, more work ought to be done to extract explicit results from Third Way path integrals of the form \eqref{eq:localisedThirdWayPartitionFunction}, including a calculation of the 1-loop determinant therein. We hope our work paves a new viable path forward for non-perturbative calculations in  gauge theory.

\section*{Acknowledgements}
ASA is happy to acknowledge insightful interactions with Victor Lekeu, Ingmar Saberi, Lewis Cole, Christoph Uhlemann, Daniel Thompson, Saskia Demulder, Matt Roberts, and Neil Lambert. He is also grateful to the Fundamental Physics group at Chulalongkorn University, Thailand, for hospitality. He is supported by the FWO-Vlaanderen through the project G006119N, as well as by the Vrije Universiteit Brussel through the Strategic Research Program “High-Energy Physics”, and by an FWO Senior Postdoctoral Fellowship.
Likewise, DK is happy to acknowledge useful interactions with Alexandre Sevrin, Ben Craps, Alberto Mariotti, Leron Borsten and especially ASA himself with regards to this project and the master thesis leading up to it.


\appendix
\section{Notation and conventions}
Some notation:
\begin{itemize}
     \item $\zeta^\alpha,\hat \zeta_\alpha$ for the two bosonic spinors for each SUSY, \emph{with this index placement by default}. We may also use $\eta,\hat \eta$ for the same.
     \item Spinor bilinears are always written $\zeta\hat\zeta$ and $\zeta \gamma_\mu\hat\zeta$, both of which obey NorthWest/SouthEast conventions for the spinor indices. If we need to switch the roles, we will write $_{\mathcal C}$, so for example $\zeta\hat\zeta=-\hat\zeta_{\mathcal C}\zeta_{\mathcal C}$ when both are bosonic.
     \item $\mu\nu\rho\cdots$ as curved $M_3$ indices
     \item $mnpqr\cdots$ as flat $M_3$ indices
     \item $\gamma^{\mu\nu}=\frac{1}{2}\gamma^{[\mu}\gamma^{\nu]}$, similarly for $\gamma^{mn}$
     \item $M_\alpha{}^\beta$ is the index placement for any $2\times 2$ matrix $M$ acting on spinors. This includes $(\zeta\otimes \hat\epsilon)_\alpha{}^\beta=\zeta_\alpha \hat\epsilon^\beta$.
     \item For $\frak g$ generators, write $T_a\,,\quad a=1,2,\cdots \dim g$, and commutation relations $[T_a,T_b]=f_{ab}{}^c T_c$. The structure constants $f_{ab}{}^c$ are {\bf real} in our convention, so that $T_a$ are all represented by antihermitian matrices for compact $\frak g$.
          \item In the context of a Manin pair $(\frak d,\frak g,\eta)$, we let the set $\{T_a,\tilde T^b\}$ generate $\frak d$ and impose $\eta(T_a,\tilde T^b)=\delta^a_b$ as the only nonvanishing $\eta$ matrix element up to symmetry, where $\{T_a\}$ have the same commutation relations in $\frak{d}$ as they do in $\frak g$.

     For $\tilde{\frak g}$, write $\tilde T^a\,,\quad a=1,2,\cdots \dim g$, and commutation relations $[\tilde T^a,\tilde T^b]=\tilde f^{ab}{}_c \tilde T^c+\tilde h^{abc}T_c$ ($g^{abc}$ vanishes when $\tilde {\frak g}$ is a Lie algebra).

     We have $\frak{d}=\frak{g}\oplus\tilde{\frak{g}}$ as vector spaces (\emph{not} as Lie algebras) in the Manin pair/triple case.
     \item $P$ is the projector onto $\frak g$ in the situation just above; $\tilde P=(1-P)$ is the projector onto $\tilde{\frak g}$.
     \item For a generic basis of $\frak d$  we may use $\mathbb T_A$ for the generators. Similarly, the matrix coefficients of $\eta$ are $\eta(\mathbb T_A,\mathbb T_B)=\eta_{AB}$.
     \item We write $x,\tilde x,\mathbbm x$ for elements of $\frak{g},\tilde{\frak g},\frak d$.
     \item We write $A$ for a $\frak{g}$ gauge field, $\mathbb A$ for a $\frak{d}$ gauge field, $\tilde A$ for a $\tilde{\frak g}$ gauge field or else for the components of $\mathbb A$ valued in $\tilde{\frak g}$ in the Manin pair case.
     \item Derivatives --- such are the exterior derivative $\dr$ --- act from the \textbf{left}: $\dr(ab)=\dr a b + (-1)^a a \dr b$. Variations, however, act from the \textbf{right}. (This is to accord with \cite{kapustin2010exact} and to eliminate some minus signs when obtaining equations of motion and the like.)
     \item We write $\mathcal D\equiv \dr + [\mathbb A,\bullet]$ for the covariant derivative associated to $\mathbb A$, and $D_A\equiv \dr +[A,\bullet]$ for the covariant derivative associated to $A$.
     \item We write $\bbA,\bbsigma,\bblambda,\hat\bblambda,\mathbb D$ for the $N=2$ $\frak d$-valued multiplet, and $A,\sigma,\lambda,\hat\lambda, D$ for the $N=2$ $\frak g$-valued multiplet. The $\sigma,D$ are bosonic scalars and $\lambda,\hat\lambda$ are fermionic spinors (resp.~for the doublestruck versions).
     \item $M$ for the linear map $M:\frak{d}\to\frak{d}$ defining the Manin pair mass term which has $M\frak{g}=0$; $\tilde M$ for the ``dual one'' with $\tilde M\tilde{\frak g}=0$. The nonvanishing components of $M$ in the basis $\{T_a,\tilde T^a\}$ are $M^{ab}$.
\end{itemize}

\subsection{Differential form conventions for possibly fermion-valued forms}
Here we assume any signature of the metric along with any Grassmann parity of the forms: namely, given we define the components of a $p$-form by
\be
\label{eq:formcpts}
\alpha\equiv \frac{1}{p!} \dr x^{\mu_1}\cdots \dr x^{\mu_p}\alpha_{\mu_1\cdots \mu_p}
\ee
we may allow $\alpha_{\mu_1\cdots \mu_p}$ to be grassmann-odd. This implies that the components depend on the ordering, above. {\bf We assume this ordering is the canoncal ordering.} Furthermore we may define another set of coefficients $\bar\alpha_{\mu_1\cdots \mu_p} $ via
\be
\label{eq:barredformcpts}
\alpha\equiv \frac{1}{p!} \bar\alpha_{\mu_1\cdots \mu_p} \dr x^{\mu_1}\cdots \dr x^{\mu_p}\,,
\ee
which for consistency satisfy (where $F(\alpha)$ is the fermion number which is 1 when $\alpha_{\mu_1\cdots \mu_p}$ is fermionic)
\be
\bar\alpha_{\mu_1\cdots \mu_p}=(-1)^{p}F(\alpha)\alpha_{\mu_1\cdots \mu_p}\,.
\ee

We may then define a $C^\infty$-bilinear form $\langle\bullet,\bullet\rangle_p$ mapping into (possibly fermionic-valued) scalars
\be
\langle\alpha,\beta\rangle_p\equiv \frac{1}{p!} \bar\alpha_{\mu_1\cdots \mu_p}\beta^{\mu_1\cdots \mu_p}
\ee
where the index gymnastics are via any metric. If we declare
\be
\boxed{T(\alpha)\equiv F(\alpha)+ p_\alpha
}\ee
to be the total degree of a possibly fermionic $p$-form $\alpha$, we then calculate 
\be
\langle\alpha,\beta\rangle_p=(-1)^{{T(\alpha)T(\beta)+p}}\langle\beta,\alpha\rangle_p
\ee
whence this form is graded-(anti)symmetric depending on the value of $p$. We may then define the Hodge star ${\star}_p$ (where we retain the subscript $p$ for clarity for now) via
\be
\alpha{\star}_p\beta\equiv (-1)^{d F(\beta)}\langle \alpha,\beta\rangle_p {\star}_01
\ee
which entails
\be
\alpha{\star}_p\beta=(-1)^{T(\alpha)T(\beta)+d(T(\alpha)+T(\beta))+p}\beta{\star}_p\alpha=(-1)^{(T(\alpha)+d)(T(\beta)+d)+p{+d}}\beta{\star}_p\alpha\,.
\ee

Note that the definition of the star above is consistent in the sense that it is well-defined as a map of $p$-forms $\beta$ to $d-p$-forms $\star \beta$ (if implicitly so). To determine this map explicitly we parameterise
\be
{\star}_p\beta\equiv \frac{\dr x^{\nu_1}\cdots \dr x^{\nu_{d-p}}}{(d-p)!} ({\star}_p)^{\mu_1\cdots \mu_p}{}_{\nu_1\cdots \nu_{d-p}} \frac{1}{p!}\beta_{\mu_1\cdots \mu_p}
\ee
for tensors $({\star}_p)^{\mu_1\cdots \mu_p}{}_{\nu_1\cdots \nu_{d-p}}$ to be determined.

Then we calculate $\alpha{\star}_p\beta$ for which it is convenient to use $\bar \alpha_{\mu\nu\cdots}$ immediately. We get
\begin{align}
\alpha{\star}_p\beta=&\frac{1}{p!} \bar\alpha_{\mu_1\cdots \mu_p} \dr x^{\mu_1}\cdots \dr x^{\mu_p}   \frac{\dr x^{\nu_1}\cdots \dr x^{\nu_{d-p}}}{(d-p)!} ({\star}_p)^{\rho_1\cdots \rho_p}{}_{\nu_1\cdots \nu_{d-p}} \frac{1}{p!}\beta_{\rho_1\cdots \rho_p} \\
&=\frac{1}{p!} \bar\alpha_{\mu_1\cdots \mu_p} (\dr x)^d \varepsilon^{\mu_1\mu_2\cdots \mu_p\nu_1\cdots\nu_{d-p}} \frac{1}{(d-p)!} ({\star}_p)^{\rho_1\cdots \rho_p}{}_{\nu_1\cdots \nu_{d-p}} \frac{1}{p!}\beta_{\rho_1\cdots \rho_p} \\
&=(-1)^{dF(\beta)}\frac{1}{p!} \bar\alpha_{\mu_1\cdots \mu_p} \frac{1}{(d-p)!}  \frac{1}{p!}\beta_{\rho_1\cdots \rho_p} (\dr x)^d \varepsilon^{\mu_1\mu_2\cdots \mu_p\nu_1\cdots\nu_{d-p}} ({\star}_p)^{\rho_1\cdots \rho_p}{}_{\nu_1\cdots \nu_{d-p}}
\end{align}

Since (where the epsilons are the $SL$-invariant tensor densities whose values are $\pm 1$)
\be
\varepsilon^{\mu_1\cdots\mu_p\nu_{1}\cdots\nu_{d-p}}\varepsilon_{\rho_1\cdots\rho_p\nu_{1}\cdots\nu_{d-p}}=p!(d-p)! \delta^{\nu_1}_{[\rho_1}\cdots \delta^{\nu_p}_{\rho_p]}
\ee
and by definition
\be
{\star_0}1=\frac{1}{d!}\dr x^{\mu_1}\cdots \dr x^{\mu_d} {\star_0}_{\mu_1\cdots \mu_d}=\underbrace{\dr x^1\dr x^2\cdots \dr x^d}_{(\dr x)^d} ({\star_0})_{12\cdots d}\,,
\ee
if we enforce $\alpha{\star}_p\beta=(-1)^{d F(\beta)}\langle \alpha,\beta\rangle_p {\star}_01$ we obtain
\be
({\star}_p)^{\rho_1\cdots \rho_p}{}_{\nu_1\cdots \nu_{d-p}}=(\star_0)
_{\sigma_1\cdots \sigma_p\nu_1\cdots\nu_{d-p}} g^{\rho_1\sigma_1}\cdots g^{\rho_p\sigma_p}
\ee
or, more compactly
\be
{\star}_p\beta\equiv \frac{\dr x^{\nu_1}\cdots \dr x^{\nu_{d-p}}}{(d-p)!} ({\star}_0)_{\sigma_1\cdots \sigma_p\nu_1\cdots\nu_{d-p}} \frac{1}{p!}\beta^{\sigma_1\cdots \sigma_p}
\ee
where the indices on $\beta$ are raised with the metric $g$. 

For the square of the Hodge star we need the identity
\be
\det g^{-1}=\frac{1}{d!} \varepsilon_{\mu_1\cdots \mu_d}\varepsilon_{\nu_1\cdots \nu_d} g^{\mu_1\nu_1}\cdots g^{\mu_d\nu_d}
\ee
whence the index-raised $\star_0$ has
\be
{\star_0}^{\mu_1\cdots \mu_d}=\det g^{-1}((\star_0)_{12\cdots d})\varepsilon^{\mu_1\cdots \mu_d}
\ee
and then the last formula for the Hodge star yields
\be
({\star_{d-p}}{\star_p}\beta)_{\mu_1\cdots \mu_p}=\det g^{-1} ((\star_0)_{12\cdots d})^2 (-1)^{p(d-p)} \beta_{\mu_1\cdots \mu_p}
\ee
Thus
\be
(\star_0)_{12\cdots d}\equiv \sqrt{|\det g|}\implies \boxed{\star^2=(-1)^{p(d-p)}\operatorname{sign}\det g}
\ee
which is exactly the same formula as for bosonic-valued forms. 
\bigskip

In summary,
\be
\boxed{{\star}_p\beta\equiv \frac{\dr x^{\nu_1}\cdots \dr x^{\nu_{d-p}}}{(d-p)!}  \sqrt{|\det g|}\varepsilon_{\sigma_1\cdots \sigma_p\nu_1\cdots\nu_{d-p}} \frac{1}{p!}\beta^{\sigma_1\cdots \sigma_p}
}\ee
\be
\boxed{\alpha{\star}_p\beta\equiv (-1)^{d F(\beta)}\langle \alpha,\beta\rangle_p {\star}1
}\ee
\be
\boxed{\langle\alpha,\beta\rangle_p\equiv \frac{1}{p!} \bar\alpha_{\mu_1\cdots \mu_p}\beta^{\mu_1\cdots \mu_p}
}\ee
with the following signs under permutations (of $p$-forms)
\be
\boxed{\alpha{\star}_p\beta=(-1)^{(T(\alpha)+d)(T(\beta)+d)+p}\beta{\star}_p\alpha\,,
}\ee
\be
\boxed{
\langle\alpha,\beta\rangle_p=(-1)^{T(\alpha)T(\beta)+p}\langle\beta,\alpha\rangle_p\,.
}\ee

\subsubsection{Integration}

We define the integral of a $d$-form $\alpha$ in a local coordinate patch of some $d$-fold as the map $\alpha\to\int\alpha$
\be
\int \alpha=\int \dr x^1\dr x^2\cdots \dr x^d \;\frac{\varepsilon^{\mu_1\mu_2\cdots\mu_d}}{d!}\alpha_{\mu_1\mu_2\cdots\mu_d}
\ee
where the components of the form are defined via \eqref{eq:formcpts}. (The domain of the integral has been omitted.) The ordering in that formula is important in the case where the components are fermionic-valued ($F(\alpha)=1$). This expression is manifestly invariant under orientation-preserving diffeomorphisms (since $\varepsilon$ is).

If $\delta$ is a \emph{right} graded differential, as in the main text, then
\be
\delta \alpha=\frac{1}{d!}\dr x^{\mu_1} \cdots \dr x^{\mu_d} \delta \alpha_{\mu_1\mu_2\cdots\mu_d}
\ee
and thus there are no funny signs in $\delta \int\alpha=\int \delta\alpha$. \emph{However}, (now for a $p$-form)
\be
\label{nastybarredvariationsigns}
(\overline{\delta \alpha})_{\mu_1\cdots \mu_p}=(-1)^{p\delta} \delta \bar \alpha_{\mu_1\cdots \mu_p}=(-1)^{{p(\delta+F(\alpha))}}\delta \alpha_{\mu_1\cdots\mu_p}\,,
\ee
where $\delta$ in the exponent is the total parity of $\delta$.

Moreover since 
\be
(c\alpha)_{\mu_1\cdots \mu_d}=(-1)^{d F(c)}c \,\alpha_{\mu_1\cdots\mu_d}\,,\qquad (\alpha c)_{\mu_1\cdots \mu_d}= \alpha_{\mu_1\cdots \mu_d} c
\ee
(the signs are reversed for the barred components), whenever $c$ is constant we have
\be
\int c\alpha=(-1)^{d F(c)}c\int \alpha
\ee
but
\be
\int\alpha c=\Big(\int \alpha \Big)c\,.
\ee
Put briefly, we may pull things out of integrals from the right.

Finally for the variation of terms like $\alpha{\star}\beta$ where $\alpha,\beta$ are $p$-forms, we calculate $\alpha\star\beta=(-1)^{d F(\alpha)} {\star} 1 \, \langle \alpha,\beta\rangle$ whence via \eqref{nastybarredvariationsigns} (for $\delta$ a right derivative again)
\be
\delta(\alpha{\star}\beta)=\alpha{\star \delta \beta}+(-1)^{d F(\alpha) + F(\beta )\delta} {\star} 1 \, (p!)^{-1} \delta \bar \alpha_{\mu_1\cdots \mu_p}\beta^{\mu_1\cdots \mu_p}
\ee
which gives
\be
\delta(\alpha{\star}\beta)=\alpha{\star \delta \beta}+(-1)^{ (p+d+F(\beta))\delta }\delta\alpha{\star}\beta\,,
\ee
or equivalently (in terms of the total degree)
\be
\delta(\alpha{\star}\beta)=\alpha{\star \delta \beta}+(-1)^{T({\star}\beta)\delta }\delta\alpha{\star}\beta\,,.
\ee

\subsection{Spinors and Gamma Matrices}
\subsubsection{Gamma matrices}
In this part of the appendix we list our conventions regarding the gamma matrices $\gamma^m = (\gamma^m{_\alpha}{^\beta})$, as well as some useful identities. We choose the represent the gamma matrices simply by the Pauli matrices,
\begin{align}
\gamma^1 &=
\begin{pmatrix}
0 & 1\\
1 & 0
\end{pmatrix}
&
\gamma^2 &=
\begin{pmatrix}
0 & -i\\
i & 0
\end{pmatrix}
&
\gamma^3 &=
\begin{pmatrix}
1 & 0\\
0 & -1
\end{pmatrix},
\end{align}
though, for most of our discussion the specific representation isn't relevant. However, we will assume they are Hermitian, i.e. $(\gamma^m)^\dagger = \gamma^m$. It follows that they satisfy the algebra
\begin{equation}
\gamma^m\gamma^n = \delta^{mn}\idop + i\varepsilon^{mn\ell}\gamma_\ell.
\end{equation}
with $\idop = (\delta_\alpha{^\beta})$ denoting the $2\times 2$ unit matrix in spinor space. 
\subsubsection{Spinors}
Spinors will either be vectors $\hzeta = (\hzeta_\alpha)$ or their conjugates $\zeta = (\zeta^\alpha)$ of the fundamental representation of $SU(2)$. As such, under local Lorentz transformations these transform as
\begin{subequations}
\begin{align}
\delta_\lambda\hzeta 
&=
+\frac{1}{4}\lambda^{mn}\gamma_{mn}\hzeta
=
+\frac{i}{4}\varepsilon_{mn\ell}\lambda^{mn}\gamma^\ell\hzeta
\\
\delta_\lambda\zeta 
&=
-\frac{1}{4}\lambda^{mn}\zeta\gamma_{mn}
=
-\frac{i}{4}\varepsilon_{mn\ell}\lambda^{mn}\zeta\gamma^\ell
\end{align}
\end{subequations}
since Lorentz generators are represented as $\rho(M_{mn}) = \frac{1}{2}\gamma_{mn}$. One should note that Hermitian conjugation maps between vectors and their conjugates. Therefore, we write
\begin{subequations}
\begin{align}
\zeta^\dagger &= (\zeta^\dagger{_\alpha}) = ((\zeta^\alpha)^\ast)
\\
\hzeta{^\dagger} &= (\hzeta^{\dagger\alpha}) = ((\hzeta_\alpha)^\ast)
\end{align}
\end{subequations}
not to be confused with the antisymmetric Northwest-Southeast contraction by the charge conjugation matrix in the following:
\subsubsection{Charge conjugation}
Now, to move on to charge conjugation: we define the charge conjugation matrix $\Cj$ to be
\begin{align}
\Cj &= (\varepsilon^{\alpha\beta}) =
\begin{pmatrix}
0 & 1\\
-1 & 0
\end{pmatrix}
=
-i\gamma^2
&
&\Rightarrow
&
\Cj^{-1} &= (-\varepsilon_{\alpha\beta}) = -\Cj = \Cj^t
\end{align}
We use $\varepsilon_{\alpha\beta}$ and $\varepsilon^{\alpha\beta}$ to raise and lower spinorial indices in the Northwest-Southeast convention, i.e.
\begin{subequations}
\begin{align}
\hzeta_\Cj &= (\Cj\hzeta)^t{\phantom{^{-1}}} = (\varepsilon^{\alpha\beta}\hzeta_\beta) =: (\hzeta{^\alpha})
\\
\zeta_\Cj &= \Cj^{-1}\zeta^t = (\zeta^\beta\varepsilon_{\beta\alpha}) =: (\zeta_\alpha)
\end{align}
\end{subequations}
$t$ denoting the transpose, as per the column/row vector interpretation. Following this convention it further also follows that
\begin{subequations}
\begin{gather}
\delta_\alpha{^\beta} = \varepsilon_\alpha{^\beta} = -\varepsilon^\beta{_\alpha}
\\
\gamma^{m\alpha\beta} = \gamma^{m\beta\alpha}
\\
\gamma^m{_{\alpha\beta}} = \gamma^m{_{\beta\alpha}}
\end{gather}
\end{subequations}
This affects spinor bilinears in the following way:
\begin{align}
\zeta\hzeta &= -\hzeta_\Cj\zeta_\Cj
&
\zeta\gamma^m\hzeta &= \hzeta_\Cj\gamma^m\zeta_\Cj
\end{align}
for bosonic spinors, with a straightforward fermionic generalisation.
\subsubsection{Fierz identity and more on spinor bilinears}
Finally, let us comment on the Fierz identity. The Fierz identity reads
\begin{equation}
M = \frac{1}{2}\idop\tr M + \frac{1}{2}\gamma^m\tr(\gamma_mM)
\end{equation}
for some spinor space matrix $M = (M_\alpha{^\beta})$. Particularly, for spinor bilinears $\hzeta\otimes\zeta = (\hzeta_\alpha\zeta^\beta)$ this reads
\begin{align}
\hzeta\otimes\zeta &= \frac{1}{2}(\zeta\hzeta)\idop + \frac{1}{2}\slashed{K}
&
K^m &\equiv \zeta\gamma^m\hzeta.
\end{align}
again assuming $\hzeta$, $\zeta$ bosonic. This leads to useful identities such as
\begin{subequations}
\begin{align}
\hzeta\otimes\zeta - \zeta_\Cj\otimes\hzeta_\Cj &= (\zeta\hzeta)\idop
\\
\hzeta\otimes\zeta + \zeta_\Cj\otimes\hzeta_\Cj &= \slashed{K}
\end{align}
\end{subequations}
Furthermore, noting that
\begin{subequations}
\begin{align}
\gamma^m\hzeta\otimes\zeta\phantom{\gamma^m} &= \frac{1}{2}K^m\idop + \frac{1}{2}(\zeta\hzeta)\gamma^m + \frac{i}{2}\varepsilon^{mn\ell}K_n\gamma_\ell
\\
\hzeta\otimes\zeta\gamma^m &= \frac{1}{2}K^m\idop + \frac{1}{2}(\zeta\hzeta)\gamma^m - \frac{i}{2}\varepsilon^{mn\ell}K_n\gamma_\ell
\end{align}
\end{subequations}
we arrive at even more identities
\begin{subequations}
\begin{gather}
\gamma^m\hzeta\otimes\zeta\phantom{\gamma^m} + \phantom{\gamma^m}\zeta_\Cj\otimes\hzeta_\Cj\gamma^m
= K^m\idop
\\
\phantom{\gamma^m}\hzeta\otimes\zeta\gamma^m + \gamma^m\zeta_\Cj\otimes\hzeta_\Cj\phantom{\gamma^m}
=
K^m\idop
\\
\vdots\nonumber
\end{gather}
\end{subequations}
and so on, by taking the appropriate linear combinations. The reason for listing these identities is that they make for very efficient tools in computing things such as the algebra $\{\delta_\zeta,\hat\delta_{\hat\zeta}\}$, or verifying the nilpotence of evanescent supersymmetry $\delta^2 = \hat\delta^2 = 0$.

\subsection{Geometry of the round 3-sphere and spinor derivatives}
\label{appendix:sphere}
We normalise the $\frak{su}(2)$ algebra of left-invariant vector fields on $S^3$ as
\begin{equation}
[X_m,X_n] = -\frac{2}{\ell}\varepsilon_{mn}{^p}X_p
\end{equation}
where $\ell$ is the 3-sphere radius, $\varepsilon^{123}=1$ is totally antisymmetric, and flat $abc\dots$ indices are lowered with $\delta_{mn}$. The dual basis of dreibeins $e^m$ then solves 
the Maurer-Cartan equation
$\extd e^m -\varepsilon^{mnp}e_ne_p/\ell = 0$, whence the torsion-free spin connection is
\begin{equation}
\omega^{mn} = \frac{1}{\ell}\varepsilon^{mnp}e_p\,.
\end{equation}
This enters into the expression for the spinorial derivatives 
\be
\nabla\zeta \equiv \dr \zeta -\zeta\frac{1}{4}\omega^{mn}\gamma_{mn}\,,\qquad \nabla\hat\zeta\equiv \dr \hat\zeta + \frac{1}{4}\omega^{mn}\gamma_{mn}\hat\zeta
\ee
which are covariant under $\delta\zeta=-\zeta\frac{1}{4}\lambda^{mn}\gamma_{mn}$, and leads to the formulas \eqref{eq:sphereSpinorD} used in the main text. For more details we refer to e.g.~\cite[section 6.3.1]{Fan_2019}.

\newpage
\let\oldbibliography\thebibliography
\renewcommand{\thebibliography}[1]{%
    \oldbibliography{#1}%
    \setlength{\itemsep}{-1pt}%
}
\begin{multicols}{2}
{\setstretch{0}
    \small
    \bibliography{bib.bib}}
\end{multicols}
\end{document}